\documentclass[11pt]{article}

\usepackage[a4paper,margin=1in]{geometry}
\usepackage{amsmath,amssymb,amsthm,mathtools}
\numberwithin{equation}{section}
\usepackage{booktabs}
\usepackage{enumitem}
\usepackage{microtype}
\usepackage{algorithm}
\usepackage{algpseudocode}
\usepackage{xcolor}
\usepackage[hidelinks]{hyperref}
\usepackage[nameinlink,capitalise,noabbrev]{cleveref}

\usepackage{aliascnt}
\usepackage{xparse} 

\makeatletter
\let\AC@oldnewtheorem\newtheorem
\RenewDocumentCommand{\newtheorem}{s m o m o}{%
  \IfBooleanTF{#1}{%
    \AC@oldnewtheorem*{#2}{#4}%
  }{%
    \IfNoValueTF{#3}{%
      \IfNoValueTF{#5}{%
        \AC@oldnewtheorem{#2}{#4}%
      }{%
        \AC@oldnewtheorem{#2}{#4}[#5]%
      }%
    }{%
      \newaliascnt{#2}{#3}%
      \AC@oldnewtheorem{#2}[#2]{#4}%
      \aliascntresetthe{#2}%
    }%
  }%
}
\makeatother

\newtheorem{theorem}{Theorem}[section]
\newtheorem{lemma}[theorem]{Lemma}
\newtheorem{proposition}[theorem]{Proposition}

\theoremstyle{definition}
\newtheorem{definition}[theorem]{Definition}
\theoremstyle{remark}

\newcommand{\C}{\mathbb C}
\newcommand{\D}{D_{\mathrm{tr}}}
\newcommand{\TV}{d_{\mathrm{TV}}}
\newcommand{\Tr}{\operatorname{Tr}}

\newcommand{\poly}{\operatorname{poly}}
\newcommand{\diag}{\operatorname{diag}}

\newcommand{\op}{\mathrm{op}}

\newcommand{\Id}{I}

\newcommand{\ket}[1]{|#1\rangle}
\newcommand{\bra}[1]{\langle#1|}
\newcommand{\proj}[1]{|#1\rangle\!\langle#1|}
\newcommand{\abs}[1]{\left|#1\right|}
\newcommand{\norm}[1]{\left\lVert#1\right\rVert}

\title{Approximating the Trace Distance Between \\ Product Quantum States}
\author{%
Kun He\thanks{Renmin University of China. Email: \href{mailto:hekun2023@ruc.edu.cn}{\texttt{hekun2023@ruc.edu.cn}}.}
\and
Dimitrios Myrisiotis\thanks{Great Bay University. Email: \href{mailto:dimyrisiotis@gbu.edu.cn}{\texttt{dimyrisiotis@gbu.edu.cn}}.}
\and
Junhong Nie\thanks{Shandong University. Email: \href{mailto:niejunhong@sdu.edu.cn}{\texttt{niejunhong@sdu.edu.cn}.}}
\and
Zongqi Wan\thanks{Great Bay University. Email: \href{mailto:zqwan@gbu.edu.cn}{\texttt{zqwan@gbu.edu.cn}}.}}
\date{}
\begin{document}
\maketitle

\begin{abstract}
We study the trace distance
\[
  \D(\rho,\sigma)
  =\frac12\norm{\rho-\sigma}_1,
  \qquad
  \rho=\bigotimes_{i=1}^n\rho_i,\quad
  \sigma=\bigotimes_{i=1}^n\sigma_i,
\]
when the two exponentially large states are specified by their local factors.
We give a deterministic approximation within a universal constant factor for rational product inputs.
Its running time is polynomial in the number of factors, the local dimension, and the input bit length.
In the opposite direction, exact computation is $\#\mathsf P$-hard even for diagonal qubit states, by the corresponding hardness of total variation distance between product distributions.

The proof uses local Uhlmann-optimal purifications to reduce the problem to estimating the product-fidelity defect and the trace norm of a structured first-order operator.
Although this operator acts on an exponentially large space, we approximate its trace norm by a local convex surrogate that admits a polynomial-size classical conic formulation.
A square-function estimate shows that the surrogate upper-bounds this trace norm.
Conversely, duality and local dephasing reduce the reverse comparison to a head--tail inequality for independent centered random variables, showing that the surrogate is at most a dimension-free constant times the same norm.

\end{abstract}

\newpage

\section{Introduction}

Trace distance answers one of the most basic operational questions in quantum information: given two quantum states, how well can any measurement distinguish them?
For density matrices \(\rho\) and \(\sigma\),
\[
    \D(\rho,\sigma)
    =\frac12\norm{\rho-\sigma}_1
    =\max_{0\le M\le I}
      \left|\Tr(M\rho)-\Tr(M\sigma)\right|.
\]
Thus \(\D(\rho,\sigma)\) is exactly the largest statistical separation produced by a quantum measurement.
Equivalently, if the two states are given with equal prior probabilities, Helstrom's theorem gives the optimal discrimination probability~\cite{Watrous2018}:
\[
    p_{\mathrm{succ}}^\star
    =\frac12\bigl(1+\D(\rho,\sigma)\bigr).
\]
Approximating trace distance therefore means approximating an operational optimum, rather than merely evaluating a convenient surrogate for distinguishability.

For explicitly represented states, trace distance can be computed by standard spectral linear algebra.
The computational picture changes when the states are specified succinctly.
In this paper, we study the trace distance between two product states of the form
\[
    \rho=\bigotimes_{i=1}^n\rho_i,
    \qquad
    \sigma=\bigotimes_{i=1}^n\sigma_i,
\]
where each local factor is an \(m\times m\) density matrix.
The input contains only \(n\) local pairs, but the two global states act on a space of dimension \(m^n\).
Even writing down \(\rho-\sigma\), let alone diagonalizing it or constructing its positive spectral projector, requires exponential space.
This produces a particularly sharp tension: each input state is a product state and hence contains no entanglement, yet the optimal measurement comparing the two states may be collective across all \(n\) sites.
The difficulty is created by taking the difference and optimizing over global measurements, not by correlations within either input state.

The same tension appears in matrix-analytic terms.
The trace norm is the Schatten-\(1\), or nuclear, norm, a canonical unitarily invariant norm and the noncommutative analogue of the \(\ell_1\)-norm.
Nuclear-norm computation and optimization play a central role in low-rank recovery, matrix completion, and related convex formulations \cite{recht2010guaranteed,candes2012exact}.
The trace norm tensorizes on a single Kronecker product:
\[
    \left\|\bigotimes_{i=1}^n A_i\right\|_1
    =\prod_{i=1}^n\norm{A_i}_1.
\]
Subtraction before taking the norm destroys this factorization.
Indeed, the Frobenius distance remains locally computable through
\[
    \norm{\rho-\sigma}_2^2
    =\prod_i\Tr(\rho_i^2)
     +\prod_i\Tr(\sigma_i^2)
     -2\prod_i\Tr(\rho_i\sigma_i),
\]
whereas no analogous formula is available for
\[
    \left\|
        \bigotimes_{i=1}^n\rho_i
        -
        \bigotimes_{i=1}^n\sigma_i
    \right\|_1.
\]
The trace norm depends on the global positive and negative spectral subspaces of \(\rho-\sigma\), and these subspaces need not respect the tensor-product description of the inputs.
Product-state trace distance is therefore a basic example of a global operational quantity hidden inside a succinct local representation.

The commuting case shows both why approximation may be possible and why the noncommutative problem requires new ideas.
For classical product distributions, the corresponding quantity is total variation distance.
Its exact computation is \(\#\mathsf P\)-hard, but randomized and deterministic fully polynomial-time approximation schemes are known \cite{FengGuoJerrumWang2023,FengLiuLiu2024,BhattacharyyaEtAl2025}.
When every local pair \((\rho_i,\sigma_i)\) commutes, simultaneous diagonalization reduces our problem exactly to this classical setting.
The classical algorithms can organize the computation around a scalar likelihood ratio.
For noncommuting local pairs there is no shared sample space, no common eigenbasis, and no scalar likelihood ordering to sparsify; moreover, the Helstrom projector can remain genuinely global.
The question is whether the succinct product description nevertheless contains enough structure to approximate the global trace norm without constructing either the matrix or its optimal measurement.

This problem provides a clean test of whether succinct product structure can be exploited in computing quantum distinguishability.
Although trace distance has a direct operational interpretation, the represented matrices have dimension \(m^n\), and the optimal measurement may still be collective.
When the local pairs commute, the problem reduces to total variation distance between product distributions, for which efficient relative-approximation schemes are known.
The noncommuting case retains the same succinct input structure but lacks a common eigenbasis, thereby highlighting the algorithmic role of noncommutativity without correlations within either state.
An efficient relative approximation would show that exponential ambient dimension and collective measurements are not, by themselves, barriers for this structured family, whereas a hardness result specific to noncommuting product states would identify a genuine distinction from the commuting case.

\subsection{Our Results}

Let each \(\rho_i,\sigma_i\) be an \(m\)-dimensional density matrix whose entries have rational real and imaginary parts, each encoded using at most \(B\) bits.
We prove the following two results.

\begin{theorem}[Informal statement of the main theorem]
\label{thm:intro-main}
There is a deterministic algorithm running in \(\poly(n,m,B)\) time that outputs a rational number \(\widehat D\) satisfying
\[
  \frac{1}{4.73}\D(\rho,\sigma)
  \le \widehat D
  \le 4.73\,\D(\rho,\sigma).
\]
The output is zero exactly when \(\rho=\sigma\).
\end{theorem}

\begin{theorem}[Informal hardness statement]
\label{thm:intro-hardness}
Exact computation of \(\D(\rho,\sigma)\) is \(\#\mathsf P\)-hard under polynomial-time Turing reductions, even when every local state is a rational diagonal qubit state.
\end{theorem}

Together, these theorems locate the problem between exact computation and fine relative approximation.
Exact evaluation is already \(\#\mathsf P\)-hard in the classical diagonal subclass, yet arbitrary noncommuting product inputs admit a deterministic polynomial-time approximation with a universal constant factor.
Although \(\D(\rho,\sigma)\le 1\), the distance can be exponentially small in the succinct input size.
A multiplicative approximation must therefore remain accurate even when the distance is exponentially small, while returning zero exactly when the two states coincide.
Theorem~\ref{thm:intro-main} provides this guarantee uniformly over all inputs, including in the nearly indistinguishable regime where standard fidelity bounds lose multiplicative control.
Whether a \((1+\varepsilon)\)-relative approximation can be achieved in polynomial time remains open.

\subsection{Technical Overview}
\label{subsec:technical-overview}

Although the input consists of only \(2n\) local \(m\times m\) matrices, the difference
\[
  \rho-\sigma
  =
  \bigotimes_{i=1}^n\rho_i
  -
  \bigotimes_{i=1}^n\sigma_i
\]
acts on a space of dimension \(m^n\).
Its trace norm does not factor across sites, and the optimal distinguishing measurement may be collective.
When all local pairs commute, simultaneous diagonalization reduces the problem to total variation distance between product distributions.
In general, the local pairs need not share eigenbases, so this classical reduction is unavailable.

The root fidelity
\[
  F(a,b):=\norm{\sqrt a\sqrt b}_1
\]
is a natural starting point because it is locally computable and tensorizes.
However, the bounds
\[
  1-F(\rho,\sigma)
  \le
  \D(\rho,\sigma)
  \le
  \sqrt{1-F(\rho,\sigma)^2}
\]
have an unbounded ratio when \(F(\rho,\sigma)\) approaches one.
This is also the regime in which the trace distance may be exponentially small, so an additive approximation is insufficient.

Our algorithm is deterministic and classical, but its main structural reduction is quantum inspired.
At a high level, it replaces the global trace norm by two quantities: a fidelity defect that can be estimated from the local matrices, and the trace norm of a structured first-order operator.

\paragraph{Isolating the first-order contribution.}

For each local pair, choose Uhlmann-optimal purifications, that is, purifications having the largest possible overlap.
In suitable coordinates, the two purifications consist of the same common component and opposite local deviations.
Their tensor-product expansion can therefore be grouped according to the number of sites at which a deviation occurs.
Let \(X_1\) denote the contribution involving a single local deviation.

Uhlmann optimality relates the total contribution of all higher-order terms to the global fidelity defect.
We prove
\[
  \norm{(\rho-\sigma)-X_1}_1
  \le
  4\bigl(1-F(\rho,\sigma)\bigr).
\]
Since fidelity tensorizes, the quantity \(1-F(\rho,\sigma)\) is determined by the local fidelities and can be evaluated to certified precision from the local data.
Together with the companion bound \(\norm{X_1}_1\le3\norm{\rho-\sigma}_1\) and the Fuchs--van de Graaf inequalities, the truncation estimate gives
\[
  \D(\rho,\sigma)
  \le
  2\bigl(1-F(\rho,\sigma)\bigr)
  +\frac12\norm{X_1}_1
  \le
  5\D(\rho,\sigma).
\]
Thus fidelity controls the discarded higher-order terms, while \(X_1\) retains the first-order information that fidelity alone misses.

\paragraph{Estimating the first-order term from local data.}

The operator \(X_1\) still acts on an exponentially large space.
After factoring out a locally computable scalar, however, it becomes a sum of \(n\) structured terms: each term changes one site while the remaining sites carry local reference states.
Bounding these terms separately would lose the cancellations between them.

Our main analytic result constructs a convex objective involving only the local matrices. Its design follows a simple head--tail principle: large local deviations are charged directly in trace norm, while the remaining fluctuations are aggregated through a quadratic cost. We prove that, after multiplying its optimal value by the locally computable scalar factor extracted above, the resulting quantity approximates \(\norm{X_1}_1\) within the universal factor \(4+4\log 2\), independently of both \(n\) and \(m\).

One direction follows from a matrix square-function estimate.
For the more difficult reverse direction, duality selects local bases in which dephasing reduces the noncommutative comparison to an inequality for independent centered scalar variables.
A scalar head--tail estimate supplies the universal constant, and contractivity of dephasing transfers the result back to the original operator.
Dephasing is used only in this proof and is not implemented by the algorithm.

Because the convex objective contains only local \(m\times m\) variables, it can be evaluated by a polynomial-size classical conic program.
The exponentially large operator \(X_1\) is never constructed.

\paragraph{Handling singular inputs and completing the algorithm.}

The local reference states arising from the purification construction may be singular.
Exact comparison of the rational input matrices first determines whether the two product states are identical.
Otherwise, it provides an explicit positive lower bound on their trace distance.
We use this bound to choose a small amount of maximally mixed noise to add to each local state.

The perturbation changes the trace distance by only a fixed relative amount, while giving the relevant local states explicit minimum-eigenvalue bounds.
Consequently, all local matrix computations and the conic program require only polynomially many bits of precision, even when the local dimension is part of the input.

On the regularized instance, the algorithm estimates the fidelity defect from the local fidelities and the first-order norm through the local conic program.
Combining these estimates, accounting for regularization and finite-precision errors, and applying a fixed rescaling yields the symmetric approximation factor \(4.73\).

All exponentially large operators in this argument are used only analytically.
The implemented algorithm uses local linear algebra, certified arithmetic, and polynomial-size classical conic optimization.

\subsection{Related Works}

\paragraph{Classical product distributions.}

The classical analogue of our problem is total variation distance between product distributions.
Feng, Guo, Jerrum, and Wang gave a randomized polynomial-time relative approximation \cite{FengGuoJerrumWang2023}.
Feng, Liu, and Liu subsequently obtained a deterministic FPTAS based on likelihood-ratio sparsification \cite{FengLiuLiu2024}.
Bhattacharyya et al.\ proved that exact evaluation is \(\#\mathsf P\)-hard \cite{BhattacharyyaEtAl2025}.
Through the diagonal embedding, the latter result immediately gives our exact hardness theorem.
The approximation algorithms, by contrast, rely on scalar likelihood-ratio structure that has no direct noncommutative counterpart.

\paragraph{Trace distance and efficiently computable surrogates.}

Trace distance, fidelity, optimal binary state discrimination, and Uhlmann's theorem are standard tools in quantum information; see \cite{Watrous2018} for background.
Under the root-fidelity convention
$F(\rho,\sigma)
  \triangleq \bigl\lVert \sqrt{\rho}\sqrt{\sigma}\bigr\rVert_1$,
fidelity tensorizes on product states:
\[
  F(\rho,\sigma)
  = \prod_{i=1}^n F(\rho_i,\sigma_i).
\]
It can therefore be evaluated efficiently from the local factors.
The Fuchs--van de Graaf inequalities~\cite{FuchsVanDeGraaf1999} give
\[
  1-F(\rho,\sigma)
  \le \D(\rho,\sigma)
  \le \sqrt{1-F(\rho,\sigma)^2}.
\]
These inequalities provide a natural efficiently computable benchmark, but not a uniform multiplicative approximation.
The ratio between the upper and lower bounds is
\[
  \sqrt{\frac{1+F(\rho,\sigma)}{1-F(\rho,\sigma)}},
\]
which diverges as \(F(\rho,\sigma)\to1\).
This is exactly the nearly-indistinguishable regime in which a relative approximation must be most sensitive.

\paragraph{Quantum-access algorithms.}

Trace-distance estimation has also been studied in models with quantum access to the input states.
Variational quantum algorithms for estimating trace distance and fidelity were proposed in \cite{ChenEtAl2022,RethinasamyEtAl2023}.
For states supplied through preparation circuits, Wang et al.\ developed block-encoding algorithms for additive-error estimation of trace distance \cite{WangEtAl2024}, and Wang and Zhang obtained improved dimension-independent query and sample complexities under a low-rank assumption \cite{WangZhang2024}.
Optimal query bounds are also known for pure states \cite{Wang2024}.
These results address a different input and resource model: they assume state-preparation, purification, or copy access, provide additive-error guarantees, and generally depend polynomially on the global rank.
They therefore do not give a classical polynomial-time multiplicative approximation from succinct local matrix descriptions when the distance can be arbitrarily small and the global rank can be exponential.

\paragraph{Matrix and tensor methods.}

For explicitly accessed matrices, randomized matrix-function methods estimate spectral sums and Schatten norms \cite{UbaruChenSaad2017}, while streaming and sketching algorithms provide space--approximation tradeoffs for the Schatten-\(1\) norm \cite{LiWoodruff2017}.
Their complexity is measured in the ambient matrix dimension or in the number of entries, matrix--vector products, or streaming updates; the ambient dimension here is \(m^n\).
Closer to the present input model, tensor-network methods approximate global functions and trace norms of matrix product operators \cite{AugustBanuls2018,LeeMoon2026}.
These methods do not presently provide worst-case polynomial bounds on the required bond dimension, variational optimization, and multiplicative error for the restricted difference of two product states considered here.
Finally, the Johnson--Schechtman disjointification theorem for sums of independent random variables \cite{JohnsonSchechtman1989} provides the analytic background for the centered scalar head--tail estimate used in this paper.

\subsection{Organization}

\Cref{sec:prelim} gives the required preliminaries.
\Cref{sec:constant} proves the \(4.73\)-approximation.
\Cref{sec:hardness} proves exact \(\#\mathsf P\)-hardness.
\Cref{sec-conclusion} concludes the paper.

\section{Preliminaries}
\label{sec:prelim}


\subsection{Notions and Tools for Quantum Information}

In this paper, all Hilbert spaces are finite-dimensional.
For density matrices $a,b$, their trace distance and fidelity are defined as follows,
\[
  \D(a,b)=\frac12\norm{a-b}_1,
  \qquad
  F(a,b)=\norm{\sqrt a\sqrt b}_1.
\]

The following folklore inequalities build a relationship between these two quantities.

\begin{lemma}[Fuchs-van de Graaf inequalities~\cite{FuchsVanDeGraaf1999}]
\label{lem:fvdg}
For any density matrices $a,b$, it holds
\begin{equation}
  1-F(a,b)\le \D(a,b)
  \le\sqrt{1-F(a,b)^2}.
  \label{eq:fvdg}
\end{equation}
\end{lemma}

It will be convenient to represent purifications by matrices.
Fix orthonormal bases for a system and an equally large auxiliary system.
For a matrix $A\in\mathbb C^{m\times m}$, define its vectorization by
\[
  \ket{A}
  =
  \sum_{x,y=1}^m A_{xy}\ket{x}\ket{y}.
\]
Directly taking the partial trace over the second register gives $\Tr_{\mathrm{aux}}\!\left(\ket{A}\!\bra{A}\right)=AA^*.$ We therefore call $A$ a \emph{purification matrix} of a density matrix $a$ when $AA^*=a$.
For two such matrices, vectorization also gives
\[
  \langle A|B\rangle
  =
  \sum_{x,y}\overline{A_{xy}}B_{xy}
  =
  \Tr(A^*B).
\]

We use the following matrix form of Uhlmann's theorem.

\begin{lemma}[Uhlmann's theorem in matrix form~\cite{Watrous2018}]
\label{lem:uhlmann-matrix}
For density matrices $a,b\in\mathbb C^{m\times m}$,
\[
  F(a,b)
  =
  \max_{\substack{AA^*=a\\BB^*=b}}
  \abs{\Tr(A^*B)}.
\]
\end{lemma}

We call a pair $(A,B)$ \emph{fidelity-attaining} when it attains this maximum.
Thus ``fidelity-attaining'' is a property of the pair: it means that the two represented purifications have the largest possible overlap.

We will repeatedly use two other standard facts.
First, trace distance is invariant under isometries and cannot increase under a quantum channel.
Partial trace therefore gives
\begin{equation}
  \D\left(\bigotimes_i\rho_i,\bigotimes_i\sigma_i\right)
  \ge \D(\rho_j,\sigma_j)
  \qquad\text{for every }j.
  \label{eq:partial-lower}
\end{equation}
Second, fidelity tensorizes exactly:
\begin{equation}
  F\left(\bigotimes_i\rho_i,\bigotimes_i\sigma_i\right)
  =\prod_iF(\rho_i,\sigma_i).
  \label{eq:fidelity-product}
\end{equation}

For the computational statements, we use the standard bit model.
The real and imaginary parts of every input entry are rational, and $B$ bounds the binary encoding length of each numerator and denominator.
The explicit input length is therefore $O(nm^2B)$; in particular, both $m$ and $B$ are part of the input.

\subsection{Disjointification Estimate}
Our final preliminary ingredient is probabilistic.
In the lower-bound argument of \cref{sec:constant}, a bounded observable will turn the local operators into independent centered scalar random variables.
We will then need to compare the expected absolute value of their sum with a quantity that separates across sites.

For independent centered scalar random variables $\xi_1,\ldots,\xi_n$, define
\begin{equation}
  \kappa((\xi_i))
  =
  \inf_{\substack{\xi_i=a_i+b_i\\
                   \mathbb Ea_i=\mathbb Eb_i=0}}
  \left\{
    \sum_i\mathbb E\abs{a_i}
    +
    \left(\sum_i\mathbb E b_i^2\right)^{1/2}
  \right\}.
  \label{eq:scalar-k-functional}
\end{equation}
The infimum splits each variable into an $L_1$ part and an $L_2$ part.

The next theorem is a centered specialization of the Johnson--Schechtman disjointification theorem \cite{JohnsonSchechtman1989}.
That theorem is stated as a qualitative equivalence, with the universal equivalence constants left implicit.
Because the constant propagates into our final approximation factor, we record the specialization needed here with the explicit value $\theta \triangleq 4+4\log 2$, and we provide a self-contained proof that tracks this constant, which is deferred to \Cref{app:proof-centered-disjointification}.

\begin{theorem}[Centered scalar disjointification]
\label{thm:centered-disjointification}
Define the universal constant $\theta \triangleq4+4\log 2$.
For every finite sequence of independent centered scalar random variables,
\begin{equation}
  \kappa((\xi_i))
  \le
  \theta\cdot 
  \mathbb E\abs{\sum_i\xi_i}.
  \label{eq:centered-disjointification}
\end{equation}
\end{theorem}

\begin{proposition}[Centered scalar pairing bound]
\label{prop:centered-scalar-pairing}
Let $(\xi_i,\eta_i)$, $i=1,\ldots,n$, be pairs of centered scalar random variables such that $\xi_1,\ldots,\xi_n$ are independent.
Suppose
\[
  \operatorname*{ess\,sup}\eta_i
  -\operatorname*{ess\,inf}\eta_i
  \le2
  \quad\text{for every }i,
  \qquad
  \sum_i\mathbb E\eta_i^2\le1.
\]
Then
\begin{equation}
  \abs{\sum_i\mathbb E(\xi_i\eta_i)}
  \le
  \theta\,\mathbb E\abs{\sum_i\xi_i}.
  \label{eq:centered-scalar-pairing}
\end{equation}
\end{proposition}

\begin{proof}
For any centered decomposition $\xi_i=a_i+b_i$, let $c_i$ be the midpoint of the essential range of $\eta_i$.
Since $\mathbb Ea_i=0$ and $\abs{\eta_i-c_i}\le1$ almost surely,
\[
\begin{aligned}
  \abs{\sum_i\mathbb E(\xi_i\eta_i)}
  &\le
  \sum_i\mathbb E\!\left(\abs{a_i}\abs{\eta_i-c_i}\right)
  +\abs{\sum_i\mathbb E(b_i\eta_i)}\\
  &\le
  \sum_i\mathbb E\abs{a_i}
  +\left(\sum_i\mathbb E b_i^2\right)^{1/2}
   \left(\sum_i\mathbb E\eta_i^2\right)^{1/2}\\
  &\le
  \sum_i\mathbb E\abs{a_i}
  +\left(\sum_i\mathbb E b_i^2\right)^{1/2}.
\end{aligned}
\]
Taking the infimum over centered decompositions and applying \cref{thm:centered-disjointification} proves the claim.
\end{proof}

\section{A Constant-Factor Approximation Algorithm}
\label{sec:constant}

The structural reason for the approximation was described in \Cref{subsec:technical-overview}.
We now give a complete specification of the algorithm.
In particular, all quantities appearing below are computed from the local input matrices, and the only optimization problem has polynomial size.
Exact symbols are used to describe the target local data; the implementation uses certified dyadic approximations with the error budgets stated below.

\paragraph{Algorithm construction.}

Below, we describe our algorithm in detail.
We also present a pseudo-code of our algorithm in \Cref{alg:product-trace-distance}.

\begin{algorithm}[t]
\caption{Constant-factor approximation of product-state trace distance}
\label{alg:product-trace-distance}
\begin{algorithmic}[1]
\Require Rational local density matrices
  $\{\rho_i,\sigma_i\}_{i=1}^n\subseteq\C^{m\times m}$
\Ensure An estimate $\widehat D$
\If{$\rho_i=\sigma_i$ for every $i$, by exact rational comparison}
  \State \Return $0$
\EndIf
\State Compute $L$ by the entrywise maximum in \eqref{eq-def-l}
\State Set
  $\beta\gets10^{-3}$,
  $\tau\gets L/(2000n)$,
  $\mu\gets\tau/m$,
  $p\gets\lceil\log_2(16n/(\beta L^2))\rceil$, and
  $\eta\gets\beta L/40$
\State Set
  $\bar\rho_i\gets(1-\tau)\rho_i+\tau\Id/m$ and
  $\bar\sigma_i\gets(1-\tau)\sigma_i+\tau\Id/m$
  for every $i$
\For{$i=1,\ldots,n$}
  \State Compute certified approximations to the fidelity-attaining data
    $(A_i,B_i,f_i,\Omega_i,x_i)$ in
    \eqref{eq:algorithm-local-data}
  \State Compute $\widetilde f_i$ with
    $\abs{\widetilde f_i-f_i}\le2^{-p}$ and set
    $\underline f_i\gets\max\{0,\widetilde f_i-2^{-p}\}$
  \State Set
    $\overline f_i\gets\min\{1,\widetilde f_i+2^{-p}\}$ and
    $\widehat q_i\gets(1+\overline f_i)/2$
  \State Compute rational $\widetilde\Omega_i,\widetilde x_i$ using
    sufficiently many guard bits
\EndFor
\State Refine the local data until tensor-product telescoping certifies
  $\widetilde\Omega_i\succeq(\mu/8)\Id$,
  $\max_i\norm{\widetilde x_i}_1\ge L/2$, and
  $\norm{\widetilde Y-Y}_1\le\eta$
\State Solve \eqref{eq:uniform-conic} with
  $(\widetilde\Omega_i,\widetilde x_i)$ to obtain
  $\operatorname{OPT}\le\widetilde N
   \le\operatorname{OPT}+\beta L/8$
\State Set
  $\widehat\delta\gets1-\prod_i\underline f_i$,
  $\widehat g\gets\prod_i\widehat q_i$, and
  $\widehat N_1\gets\widehat g(\widetilde N+\eta)$
\State Set $\widehat T\gets4\widehat\delta+\widehat N_1$
\State \Return
  $\widehat D\gets(10583/100000)\widehat T$
\end{algorithmic}
\end{algorithm}

\noindent \underline{\textit{Step 1: zero detection and regularization.}}
The algorithm first checks, using exact rational arithmetic, whether $\rho_i=\sigma_i$ for every $i$.
If so, it returns zero.
Otherwise, it computes the positive rational number
\begin{equation}
  L
  \triangleq
  \frac12
  \max_{\substack{1\le i\le n\\1\le a,b\le m}}
  \left\{
    \abs{\Re((\rho_i-\sigma_i)_{ab})},
    \abs{\Im((\rho_i-\sigma_i)_{ab})}
  \right\}.
  \label{eq-def-l}
\end{equation}
This quantity act as a regularization parameter.
Set $\tau \triangleq \frac{L}{2000n}$, and add a white noise regularization on $\rho_i$ and $\sigma_i$, resulting
\[
  \bar\rho_i=(1-\tau)\rho_i+\tau\frac{\Id}{m},
  \qquad
  \bar\sigma_i=(1-\tau)\sigma_i+\tau\frac{\Id}{m}.
\]
Then, we work on the regularized states.

\vspace{5pt}

\noindent \underline{\textit{Step 2: local purification data and the fidelity estimate.}}
Fix $\beta=10^{-3}$ and $p=\left\lceil\log_2(16n/(\beta L^2))\right\rceil$.
For every $i$, compute certified approximations to a fidelity-attaining pair of purification matrices for $(\bar\rho_i,\bar\sigma_i)$ according to \Cref{lem:uniform-bures-data}.
Denote the purification matrices of $\bar\rho_i,\bar\sigma_i$ by $A_i,B_i$, where $A_iA_i^*=\bar\rho_i$, $B_iB_i^*=\bar\sigma_i$, and $A_i^*B_i\succeq0$, and define
\begin{equation}
\begin{aligned}
  f_i&=\Tr(A_i^*B_i)=F(\bar\rho_i,\bar\sigma_i),&
  c_i^2&=\frac{1+f_i}{2},\\
  \Omega_i&=\frac{(A_i+B_i)(A_i+B_i)^*}{2(1+f_i)},&
  x_i&=\frac{\bar\rho_i-\bar\sigma_i}{c_i^2}.
\end{aligned}
\label{eq:algorithm-local-data}
\end{equation}
Compute a dyadic number $\widetilde f_i$ satisfying $\abs{\widetilde f_i-f_i}\le2^{-p}$.
Set $\underline f_i=\max\{0,\widetilde f_i-2^{-p}\}$, $\overline f_i=\min\{1,\widetilde f_i+2^{-p}\}$, and $\widehat q_i=(1+\overline f_i)/2$.
After processing all sites, set $\widehat\delta=1-\prod_i\underline f_i$ and $\widehat g=\prod_i\widehat q_i$.

\vspace{5pt}

\noindent \underline{\textit{Step 3: the first-order conic program.}}
For the precision bookkeeping, set $\eta=\beta L/40$ and define $\Omega_{\ne i}=\bigotimes_{j\ne i}\Omega_j$ and $Y=\sum_i\Omega_{\ne i}\otimes x_i$.
Compute rational density matrices $\widetilde\Omega_i$ and rational Hermitian traceless matrices $\widetilde x_i$, refining their precision as follows.
Write $\widetilde\Omega_{\ne i}=\bigotimes_{j\ne i}\widetilde\Omega_j$ and $\widetilde Y=\sum_i\widetilde\Omega_{\ne i}\otimes\widetilde x_i$.
Refine the local data until tensor-product telescoping certifies $\widetilde\Omega_i\succeq(\mu/8)\Id$ for every $i$, $\max_i\norm{\widetilde x_i}_1\ge L/2$, and $\norm{\widetilde Y-Y}_1\le\eta$.
Using these rational local matrices, solve the conic program
\begin{equation}
\begin{aligned}
 \text{minimize}\quad&
   \sum_{i=1}^n\Tr(R_i+S_i)+t,\\
 \text{subject to}\quad&
   \widetilde x_i
   =R_i-S_i
    +\frac12(\widetilde\Omega_i a_i+a_i\widetilde\Omega_i),
   && i=1,\ldots,n,\\
 & \Tr(\widetilde\Omega_i a_i)=0,\qquad
   R_i,S_i\succeq0,
   && i=1,\ldots,n,\\
 & \left(
     \sum_{i=1}^n\Tr(\widetilde\Omega_i a_i^2)
   \right)^{1/2}\le t,
\end{aligned}
\label{eq:uniform-conic}
\end{equation}
over Hermitian matrices $R_i,S_i,a_i$ and a real scalar $t$.
Denote its optimum by $\operatorname{OPT}$.
Compute a certified rational $\widetilde N$ satisfying $\operatorname{OPT}\le\widetilde N\le\operatorname{OPT}+\beta L/8$.
Set $\widehat N_Y=\widetilde N+\eta$ and $\widehat N_1=\widehat g\,\widehat N_Y$.

\vspace{5pt}

\noindent \underline{\textit{Step 4: assembly.}}
Finally, set $\widehat T=4\widehat\delta+\widehat N_1$ and return $\widehat D=(10583/100000)\widehat T$.

\paragraph{Analysis roadmap.}
The subsequent sections give the guarantee for different steps of our algorithm.
\Cref{sec:zero detection} prove that the regularization will not affect the trace distance much while providing the strict eigenvalue lower bound required to implement the inverse matrix.
\Cref{sec:local computations} showing the computations of the local state quantities used by our algorithm can be done within polynomial time.
\Cref{sec-first-order-purifications} bounds the trace distance with first-order term $X_1$ and the quantity $\hat{\delta} = 1-\prod_i f_i$.
\Cref{sec: estimating first-order} then use the conic program to estimate the first-order part.
\Cref{sec:overall} put all things together and analyze the overall approximation factor.

\subsection{Zero Detection and Regularization}
\label{sec:zero detection}

A multiplicative approximation must recognize the case $\D(\rho,\sigma)=0$ exactly.
Product structure makes this possible using only the local inputs.
It also supplies a positive scale against which we set the numerical precision in the nonzero case.

\begin{lemma}
\label{lem:zero-scale}
Let $\rho=\bigotimes_i\rho_i$ and $\sigma=\bigotimes_i\sigma_i$ be rational product density matrices whose local entries have bit length at most $B$.
Then
\[
  \rho=\sigma
  \quad\Longleftrightarrow\quad
  \rho_i=\sigma_i\text{ for every }i.
\]
If $\rho\ne\sigma$, define
\begin{equation}
  L
  \triangleq
  \frac12
  \max_{\substack{1\le i\le n\\1\le a,b\le m}}
  \left\{
    \abs{\Re((\rho_i-\sigma_i)_{ab})},
    \abs{\Im((\rho_i-\sigma_i)_{ab})}
  \right\},
\end{equation}
and let $k$ be a site at which the maximum is attained.
Then $L$ is a positive rational computable in polynomial time, and
\begin{align}\label{eq-twominusobleql-leqtd}
  2^{-O(B)}
  \le L
  \le \D(\rho_k,\sigma_k)
  \le \D(\rho,\sigma).
\end{align}
\end{lemma}

The proof is deferred to \Cref{app:proof-zero-scale}.

The constructions used later involve inverse square roots.
The next lemma mixes the same small amount of white noise into every local state and records the two properties needed below.

\begin{lemma}[Regularization]
\label{lem:regularization}
For product density matrices $\rho=\bigotimes_i\rho_i$ and $\sigma=\bigotimes_i\sigma_i$, and any $0<\tau<1$, define
\[
  \bar\rho_i=(1-\tau)\rho_i+\tau\frac{\Id}{m},
  \qquad
  \bar\sigma_i=(1-\tau)\sigma_i+\tau\frac{\Id}{m},
\]
and let
\[
  \bar\rho=\bigotimes_i\bar\rho_i,
  \qquad
  \bar\sigma=\bigotimes_i\bar\sigma_i.
\]
Then
\begin{equation}
  0
  \le
  \D(\rho,\sigma)-\D(\bar\rho,\bar\sigma)
  \le2n\tau,
  \label{eq:regularization-error}
\end{equation}
and
\[
  \bar\rho_i,\bar\sigma_i\succeq\frac{\tau}{m}\Id
  \quad\text{for every }i.
\]
\end{lemma}

The proof is deferred to \Cref{app:proof-regularization}.

\subsection{Local Matrix Computations}
\label{sec:local computations}
Our algorithm only performs numerical linear algebra on the local $m\times m$ matrices.
The next lemma packages the local data needed by the approximation algorithm.
Besides computability, the important point is the lower bound on the bisector state $\Omega$: after regularization, the local objects remain uniformly away from singularity.

\begin{lemma}
\label{lem:uniform-bures-data}
Let $a,b$ be rational $m\times m$ density matrices satisfying $a,b\succeq\mu\Id$.
There exist fidelity-attaining purification matrices $A$ for $a$ and $B$ for $b$ such that, writing
\[
  f\triangleq F(a,b),
  \qquad
  \Omega\triangleq\frac{(A+B)(A+B)^*}{2(1+f)},
\]
we have
\begin{equation}
  A^*B\succeq0,
  \qquad
  \Omega\succeq\frac{\mu}{4}\Id.
  \label{eq:bisector-gap}
\end{equation}
Moreover, for every integer $p\ge1$, a deterministic algorithm running in $\poly\bigl(m,B,p,\log(1/\mu)\bigr)$ time outputs dyadic rational approximations $\widetilde A,\widetilde B,\widetilde f,\widetilde\Omega$ satisfying
\[
  \max\left\{
    \norm{\widetilde A-A}_{\mathrm{op}},
    \norm{\widetilde B-B}_{\mathrm{op}},
    \abs{\widetilde f-f},
    \norm{\widetilde\Omega-\Omega}_{\mathrm{op}}
  \right\}
  \le 2^{-p}.
\]
\end{lemma}


The proof is deferred to \Cref{app:proof-uniform-bures-data}.

\subsection{The First-Order Part of Optimal Product Purifications}
\label{sec-first-order-purifications}

The constructions in this subsection are parametric in the chosen product pair.
In the proof of \Cref{thm:main}, they will be instantiated with the regularized pair $(\bar\rho,\bar\sigma)$ constructed there.

For each local pair $(\rho_i,\sigma_i)$, choose Uhlmann-optimal purifications $\ket{r_i},\ket{s_i}\in\C^m\otimes\C^m$ such that
\[
  \Tr_E\proj{r_i}=\rho_i,\qquad
  \Tr_E\proj{s_i}=\sigma_i,\qquad
  \langle r_i|s_i\rangle=f_i\triangleq F(\rho_i,\sigma_i)\ge0.
\]
Such a pair exists by \Cref{lem:uhlmann-matrix}, after adjusting a global phase.
Define the product purifications
\[
\begin{aligned}
  \ket{\Psi}&=\bigotimes_i\ket{r_i},
  &\Tr_E\proj{\Psi}&=\rho,\\
  \ket{\Phi}&=\bigotimes_i\ket{s_i},
  &\Tr_E\proj{\Phi}&=\sigma.
\end{aligned}
\]

We now express each local pair in bisector coordinates.
When $f_i<1$, define
\begin{equation}
  \ket{0_i}
    =\frac{\ket{r_i}+\ket{s_i}}{\sqrt{2(1+f_i)}},
  \qquad
  \ket{1_i}
    =\frac{\ket{r_i}-\ket{s_i}}{\sqrt{2(1-f_i)}}.
  \label{eq:bisector-basis}
\end{equation}
The two vectors are orthonormal.
Put
\begin{align}\label{eq:bisector-data}
  c_i=\sqrt{\frac{1+f_i}{2}},\qquad
  s_i=\sqrt{\frac{1-f_i}{2}},\qquad
  \Omega_i=\Tr_E\proj{0_i}.
\end{align}
If $f_i=1$, then $\rho_i=\sigma_i$; we take $\ket{0_i}=\ket{r_i}=\ket{s_i}$ and omit the zero $s_i\ket{1_i}$ term.
Thus

\begin{equation}
  \ket{\Psi}
  =\bigotimes_i(c_i\ket{0_i}+s_i\ket{1_i}),
  \qquad
  \ket{\Phi}
  =\bigotimes_i(c_i\ket{0_i}-s_i\ket{1_i}).
  \label{eq:product-purifications}
\end{equation}
\begin{definition}[First-order part]
\label{def:first-order-part}
For the fixed Uhlmann-optimal product purifications above, let $P_{\le1}$ project onto the span of the vacuum $\bigotimes_i\ket{0_i}$ and the vectors with exactly one local $\ket{1_i}$.
Their first-order reduced difference is
\begin{equation}
  X_1
  =
  \Tr_E\!\left[
    P_{\le1}
    \bigl(\proj{\Psi}-\proj{\Phi}\bigr)
    P_{\le1}
  \right].
  \label{eq:x1-definition}
\end{equation}
\end{definition}

The following theorem bounds the trace distance using the local fidelities $f_i$ and the first-order operator $X_1$.

\begin{theorem}
\label{thm:first-order-bridge}
For arbitrary local density matrices with fixed local Uhlmann-optimal purification pairs, let $X_1$ be their first-order part from \cref{def:first-order-part}, and let $\delta_F=1-\prod_i f_i$.
Then
\begin{equation}
\begin{aligned}
  \norm{\rho-\sigma}_1
  &\le \norm{X_1}_1+4\delta_F,\\
  \norm{X_1}_1
  &\le3\norm{\rho-\sigma}_1,\\
  \delta_F
  &\le\frac12\norm{\rho-\sigma}_1.
\end{aligned}
  \label{eq:first-order-bridge}
\end{equation}
\end{theorem}

\begin{proof}
For $K\subseteq\{i:f_i<1\}$, let $\ket K$ denote the corresponding excitation-set vector and put
\[
  \alpha_K
  =
  \prod_{i\in K}s_i\prod_{j\notin K}c_j.
\]
Expanding \eqref{eq:product-purifications} gives
\[
  \ket{\Psi}=\sum_K\alpha_K\ket K,
  \qquad
  \ket{\Phi}=\sum_K(-1)^{|K|}\alpha_K\ket K.
\]
The total squared coefficient of the terms with at least two excitations satisfies
\[
\begin{aligned}
  \sum_{|K|\ge2}\alpha_K^2
  &=
  1-\left(\prod_i c_i^2\right)
  \left(1+\sum_i\frac{s_i^2}{c_i^2}\right)\\
  &\le
  \left(1-\prod_i c_i^2\right)^2
  \le
  \delta_F^2.
\end{aligned}
\]
Here the first inequality uses $1-\prod_i c_i^2\le\sum_i s_i^2 \le\sum_i s_i^2/c_i^2$, and the second uses $f_i=c_i^2-s_i^2\le c_i^2$.

Write $P=P_{\le1}$ and $Q=\Id-P$.
For $\ket\xi\in\{\ket\Psi,\ket\Phi\}$,
\[
\begin{aligned}
  \norm{\proj\xi-P\proj\xi P}_1
  &=
  \norm{\ket{Q\xi}\!\bra{\xi}
        +\ket{P\xi}\!\bra{Q\xi}}_1\\
  &\le
  2\norm{Q\ket\xi}
  =
  2\left(\sum_{|K|\ge2}\alpha_K^2\right)^{1/2}.
\end{aligned}
\]
Using \eqref{eq:x1-definition}, trace-norm contractivity under partial trace, and the triangle inequality, we obtain
\begin{equation}
  \norm{(\rho-\sigma)-X_1}_1
  \le4\delta_F.
  \label{eq:x1-error-fidelity}
\end{equation}
The Fuchs--van de Graaf inequality in \cref{lem:fvdg} and fidelity tensorization in \eqref{eq:fidelity-product} give
\[
  \delta_F\le\frac12\norm{\rho-\sigma}_1.
\]
Finally, \eqref{eq:x1-error-fidelity} and the triangle inequality give
\[
  \norm{\rho-\sigma}_1
  \le\norm{X_1}_1+4\delta_F,
  \qquad
  \norm{X_1}_1
  \le\norm{\rho-\sigma}_1
      +4\delta_F
  \le3\norm{\rho-\sigma}_1.
\]
\end{proof}

\subsection{Estimating the First-Order Part}
\label{sec: estimating first-order}
We first study a general first-order product operator.
Let $\Omega_i\succ0$ be density matrices and let $x_i$ be Hermitian traceless matrices.
Define
\begin{equation}
  Y
  =
  \sum_{i=1}^n\Omega_{\ne i}\otimes x_i,
  \qquad
  \Omega_{\ne i}
  =
  \bigotimes_{j\ne i}\Omega_j.
  \label{eq:degree-one-general}
\end{equation}
Denote $\mathcal M_i=\{z=z^*:\Tr z=0\}.$

To estimate $\norm Y_1$, we adapt the Johnson--Schechtman decomposition for sums of independent random variables \cite{JohnsonSchechtman1989}, in the centered form of \cref{thm:centered-disjointification}.
It splits each centered random variable into an $L_1$ part, controlled term by term, and an $L_2$ part, controlled by the square root of the total variance.
We similarly split each local tangent as $x_i=u_i+v_i$.
The $u_i$'s will be controlled in trace norm, while the $v_i$'s will be controlled by the following quantum $L_2$ quantity.
For every $v_i\in\mathcal M_i$, let $a_i$ be the unique Hermitian solution of
\[
  \frac12(\Omega_i a_i+a_i\Omega_i)=v_i.
\]
Indeed, if $\Omega_i=\sum_rp_r\proj r$, then
\[
  (a_i)_{rs}=\frac{2(v_i)_{rs}}{p_r+p_s},
\]
which is well defined because every $p_r>0$.
Moreover,
\[
  \Tr(\Omega_i a_i)=\Tr v_i=0,
  \qquad
  \Tr(v_i a_i)=\Tr(\Omega_i a_i^2)\ge0.
\]
When $\Omega_i$ and $v_i$ commute, a common eigenbasis gives $\sum_r(v_i)_{rr}^2/p_r$, the classical weighted $L_2$ cost.

For the tuple $x=(x_i)$, define
\begin{equation}
\begin{split}
  K_{\mathrm{loc}}(x;\Omega)
  =
  \inf_{\substack{x_i=u_i+v_i\\
                   u_i,v_i\in\mathcal M_i}}
  \left\{
    \sum_i\norm{u_i}_1+
    \left(\sum_i\Tr(v_i a_i)\right)^{1/2}
  \right\}.
\end{split}
\label{eq:local-k-functional}
\end{equation}

It is clear that $\norm Y_1$ is upper bounded by $K_{\mathrm{loc}}(x;\Omega)$.

\begin{theorem}
\label{thm:local-head-tail-bound}
Every first-order product operator in \eqref{eq:degree-one-general} satisfies
\[
  \norm Y_1\le K_{\mathrm{loc}}(x;\Omega).
\]
\end{theorem}

\begin{proof}
Fix a feasible decomposition $x_i=u_i+v_i$, and write
\[
  Y_u=\sum_i\Omega_{\ne i}\otimes u_i,
  \qquad
  Y_v=\sum_i\Omega_{\ne i}\otimes v_i.
\]
The $L_1$ part satisfies
\[
  \norm{Y_u}_1
  \le
  \sum_i\norm{\Omega_{\ne i}\otimes u_i}_1
  =
  \sum_i\norm{u_i}_1,
\]
because $\norm{\Omega_{\ne i}}_1=1$.

For the $L_2$ part, let $A=\sum_i a_i^{(i)}$, where $a_i^{(i)}$ acts as $a_i$ at site $i$ and as the identity elsewhere.
Then
\[
  Y_v=\frac12(\Omega A+A\Omega),
  \qquad
  \Tr(\Omega_i a_i)=0,
  \qquad
  \Tr(v_i a_i)=\Tr(\Omega_i a_i^2).
\]
The centering removes all cross terms in $\Tr(\Omega A^2)$.
Schatten Cauchy--Schwarz therefore gives
\[
\begin{aligned}
  \norm{Y_v}_1
  &\le\frac12\left(\norm{\Omega A}_1+\norm{A\Omega}_1\right)\\
  &\le\sqrt{\Tr(\Omega A^2)}
  =\left(\sum_i\Tr(v_i a_i)\right)^{1/2}.
\end{aligned}
\]
Combining the two parts gives
\[
  \norm Y_1
  \le\norm{Y_u}_1+\norm{Y_v}_1
  \le
  \sum_i\norm{u_i}_1+
  \left(\sum_i\Tr(v_i a_i)\right)^{1/2}.
\]
Taking the infimum proves $\norm Y_1\le K_{\mathrm{loc}}(x;\Omega)$.
\end{proof}

We prove the reverse bound by reducing the dual of $K_{\mathrm{loc}}$ to the centered scalar setting of \cref{prop:centered-scalar-pairing}.
A local dual witness will select a product dephasing basis; in that basis its objective becomes a pairing of independent centered scalar random variables.

\begin{theorem}
\label{thm:degree-one-reverse-bound}
Every first-order product operator in \eqref{eq:degree-one-general} satisfies
\begin{equation}
  K_{\mathrm{loc}}(x;\Omega)
  \le \theta\norm Y_1.
  \label{eq:degree-one-reverse-bound}
\end{equation}
\end{theorem}

\begin{proof}
For a Hermitian $h$, regarded as a functional on $\mathcal M_i$, define
\[
\begin{aligned}
 \alpha_i(h)
 &=
 \sup_{\substack{u\in\mathcal M_i\\ \norm u_1\le1}}
       \abs{\Tr(hu)},\\
 \beta_i(h)
 &=
 \sup_{\substack{a=a^*,\ \Tr(\Omega_i a)=0\\
                  \Tr(\Omega_i a^2)\le1}}
       \frac12\abs{\Tr\!\left(h(\Omega_i a+a\Omega_i)\right)}.
\end{aligned}
\]
Duality for the infimal convolution in \eqref{eq:local-k-functional} gives
\[
 K_{\mathrm{loc}}(x;\Omega)
 =
 \sup\Bigl\{
   \sum_i\Tr(h_ix_i):
   \max_i\alpha_i(h_i)\le1,\ 
   \sum_i\beta_i(h_i)^2\le1
 \Bigr\}.
\]
Fix a feasible tuple $(h_i)$, and diagonalize
\[
  h_i=\sum_r\lambda_i(r)\proj{r_i}.
\]
Define
\[
\begin{aligned}
  \xi_i(r)&=
  \frac{\langle r_i|x_i|r_i\rangle}
       {\langle r_i|\Omega_i|r_i\rangle},\\
  \eta_i(r)&=
  \lambda_i(r)
  -\sum_t\langle t_i|\Omega_i|t_i\rangle\lambda_i(t).
\end{aligned}
\]
Every denominator is positive because $\Omega_i\succ0$.
Draw the labels $r_i$ independently with
\[
  \Pr(r_i=r)=\langle r_i|\Omega_i|r_i\rangle.
\]
Trace-norm duality gives
\[
  \alpha_i(h_i)
  =
  \frac12\left(\max_r\lambda_i(r)-\min_r\lambda_i(r)\right).
\]
Indeed, subtracting the midpoint of the eigenvalue range gives the upper bound, and half the difference of projectors onto extreme eigenvectors gives equality.
Therefore
\[
  \max_r\eta_i(r)-\min_r\eta_i(r)
  =\max_r\lambda_i(r)-\min_r\lambda_i(r)
  =2\alpha_i(h_i)
  \le2.
\]
In the definition of $\beta_i(h_i)$, the identity part of $h_i$ drops out because $\Tr(\Omega_i a)=0$.
Weighted Cauchy--Schwarz, with equality when $a$ is proportional to $h_i-\Tr(\Omega_i h_i)\Id$, gives
\[
  \beta_i(h_i)^2
  =
  \Tr\!\left(
    \Omega_i\bigl(h_i-\Tr(\Omega_i h_i)\Id\bigr)^2
  \right)
  =
  \mathbb E\eta_i^2.
\]
By construction, $\mathbb E\eta_i=0$, while $\Tr x_i=0$ gives $\mathbb E\xi_i=0$ and
\[
  \Tr(h_ix_i)
  =\sum_r\lambda_i(r)\langle r_i|x_i|r_i\rangle
  =\mathbb E(\xi_i\eta_i).
\]
The identity above and dual feasibility give
\[
  \sum_i\mathbb E\eta_i^2
  =\sum_i\beta_i(h_i)^2
  \le1.
\]
The conditions of \cref{prop:centered-scalar-pairing} are therefore satisfied, and
\[
  \abs{\sum_i\Tr(h_ix_i)}
  =
  \abs{\sum_i\mathbb E(\xi_i\eta_i)}
  \le
  \theta\,\mathbb E\abs{\sum_i\xi_i}.
\]
Let $\mathcal D$ dephase in the product of the selected eigenbases.
For $\ket{\mathbf r}=\ket{r_1}\otimes\cdots\otimes\ket{r_n}$,
\[
  \langle\mathbf r|Y|\mathbf r\rangle
  =
  \left(\prod_j\langle r_j|\Omega_j|r_j\rangle\right)
  \sum_i\xi_i(r_i).
\]
Summing the absolute diagonal entries gives
\[
  \norm{\mathcal D(Y)}_1
  =
  \mathbb E\abs{\sum_i\xi_i}.
\]
Since dephasing is trace-norm contractive,
\[
  \abs{\sum_i\Tr(h_ix_i)}
  \le
  \theta\norm{\mathcal D(Y)}_1
  \le
  \theta\norm Y_1.
\]
This holds for every feasible dual witness.
Taking the supremum proves the theorem.
\end{proof}

By \cref{thm:local-head-tail-bound,thm:degree-one-reverse-bound}, we approximate $\norm Y_1$ by approximating $K_{\mathrm{loc}}(x;\Omega)$.

\begin{theorem}[Estimating first-order part]
\label{thm:degree-one}
Suppose every entry of the local data in \eqref{eq:degree-one-general}, together with the parameters $\gamma$ and $\lambda$, has bit length at most $B$.
Let $0<\gamma,\lambda\le1$ satisfy
\[
  \Omega_i\succeq\gamma\Id\quad\text{for every }i,
  \qquad
  \lambda\le\max_i\norm{x_i}_1.
\]
For every $0<\beta<1$, there is a deterministic algorithm running in
\[
  \poly\!\left(
    n,m,B,\log\frac1\gamma,\log\frac1\lambda,
    \log\frac1\beta
  \right)
\]
time that outputs $\widehat N$ with
\begin{equation}
  \norm Y_1
  \le \widehat N
  \le (\theta+\beta)\norm Y_1.
  \label{eq:degree-one-bound}
\end{equation}
\end{theorem}

\begin{proof}
Parameterize the tail and the Hermitian head as
\[
  v_i=\frac12(\Omega_i a_i+a_i\Omega_i),
  \qquad
  u_i=R_i-S_i.
\]
The identities
\[
  \Tr v_i=\Tr(\Omega_i a_i),\qquad
  \Tr(v_i a_i)=\Tr(\Omega_i a_i^2)
\]
turn \eqref{eq:local-k-functional} into the conic program
\begin{equation}
\begin{aligned}
 \text{minimize}\quad&
   \sum_i\Tr(R_i+S_i)+t,\\
 \text{subject to}\quad&
   x_i=R_i-S_i+\frac12(\Omega_i a_i+a_i\Omega_i),\\
 & \Tr(\Omega_i a_i)=0,\qquad
   R_i,S_i\succeq0,\\
 & \left(\sum_i\Tr(\Omega_i a_i^2)\right)^{1/2}\le t.
\end{aligned}
\end{equation}
For fixed $a_i$, the head is
\[
  u_i=x_i-\frac12(\Omega_i a_i+a_i\Omega_i).
\]
The minimum over $R_i,S_i\succeq0$ with $R_i-S_i=u_i$ is attained by the positive and negative parts
\[
  R_i=(u_i)_+=\frac{\abs{u_i}+u_i}{2},
  \qquad
  S_i=(u_i)_-=\frac{\abs{u_i}-u_i}{2},
\]
and equals $\norm{u_i}_1$.
Minimizing $t$ gives the square-root term in \eqref{eq:local-k-functional}.
Thus the value of \eqref{eq:uniform-conic} is $K_{\mathrm{loc}}(x;\Omega)$.

Use the standard rational real coordinates consisting of the diagonal
entries and the real and imaginary parts of the upper-triangular entries,
and let \(z\) concatenate the coordinates of the \(a_i\)'s.  These
coordinates satisfy
\[
  \norm z_2^2
  \le \sum_i\norm{a_i}_F^2
  \le 2\norm z_2^2,
\]
and all coefficients of the resulting quadratic form are rational.
The last constraint is the ellipsoidal cone
\[
  z^\mathsf TQz\le t^2,\qquad t\ge0,
\]
where
\[
\begin{aligned}
  z^\mathsf TQz
  &=
  \sum_i\Tr(\Omega_i a_i^2)
  \ge
  \gamma\sum_i\Tr(a_i^2)
  =
  \gamma\sum_i\norm{a_i}_F^2
  \geq
  \gamma\norm z_2^2.
\end{aligned}
\]
Here the inequality uses $\Omega_i\succeq\gamma\Id$ and $a_i^2\succeq0$.
Hence $Q\succeq\gamma\Id$, so in particular $\lambda_{\min}(Q)\ge\gamma$.
Its coefficients have polynomial bit length, so the cone has a polynomial-time weak separation oracle.
The program has $O(nm^2)$ variables, $2n$ positive-semidefinite blocks of size $m$, and one ellipsoidal cone of dimension $O(nm^2)$.

Choose a polynomial-bit $c\ge1+\max_i\norm{x_i}_\infty$.
Then
\[
  a_i=0,\qquad
  R_i=c\Id+x_i/2,\qquad
  S_i=c\Id-x_i/2,\qquad t=1
\]
is strictly feasible.
Choose a polynomial-bit $U$ larger than its objective and restrict the program to the sublevel $U$.
On this sublevel,
\[
  \sum_i\bigl(\norm{R_i}_F+\norm{S_i}_F\bigr)
  \le U,
  \qquad
  \sum_i\norm{a_i}_F^2
  \le\frac{U^2}{\gamma}.
\]
The first bound uses positivity; the second uses the ellipsoidal constraint and $\Omega_i\succeq\gamma\Id$.
The strict point and these bounds give polynomially bounded inner and outer radii.
Standard weak conic optimization therefore returns
\[
  K_{\mathrm{loc}}(x;\Omega)
  \le\widehat K
  \le K_{\mathrm{loc}}(x;\Omega)+\eta
\]
in time polynomial in $n,m,B,\log(1/\gamma)$, and $\log(1/\eta)$.

Because every $x_j$ is traceless, partial trace gives
\[
  \Tr_{\ne i}Y=x_i,
  \qquad
  \lambda\le\max_i\norm{x_i}_1\le\norm Y_1.
\]
Set $\eta=\beta\lambda$ and return $\widehat N=\widehat K$.
By \cref{thm:local-head-tail-bound,thm:degree-one-reverse-bound},
\[
  \norm Y_1
  \le\widehat N
  \le\theta\norm Y_1+\beta\lambda
  \le(\theta+\beta)\norm Y_1,
\]
which proves \eqref{eq:degree-one-bound}.
\end{proof}

\begin{theorem}
\label{thm:full-rank-first-order}
Let
\[
  \rho=\bigotimes_{i=1}^n\rho_i,
  \qquad
  \sigma=\bigotimes_{i=1}^n\sigma_i
\]
be rational product states.
Suppose
\[
  \rho_i,\sigma_i\succeq\mu\Id
  \quad\text{for every }i,
  \qquad
  0<\lambda\le\max_i\norm{\rho_i-\sigma_i}_1,
\]
where the rational inputs and parameters have bit length at most $B$.
Use the Uhlmann-optimal purification pairs furnished by \cref{lem:uniform-bures-data}, and let $X_1$ be their first-order part from \cref{def:first-order-part}.

For every rational $0<\beta<1$, there is a deterministic algorithm running in
\[
  \poly\left(
    n,m,B,\log\frac1\mu,\log\frac1\lambda,\log\frac1\beta
  \right)
\]
time that outputs a rational number $\widehat N_1$ satisfying
\begin{equation}
  \norm{X_1}_1
  \le\widehat N_1
  \le(1+\beta)(\theta+\beta)\norm{X_1}_1.
  \label{eq:full-rank-X1-interface}
\end{equation}
\end{theorem}

\begin{proof}
For each $i$, let
\[
  f_i=F(\rho_i,\sigma_i),
  \qquad
  c_i^2=\frac{1+f_i}{2},
  \qquad
  x_i=\frac{\rho_i-\sigma_i}{c_i^2},
\]
and let $\Omega_i$ be the corresponding bisector state.
Define
\[
  Y=\sum_{i=1}^n\Omega_{\ne i}\otimes x_i,
  \qquad
  \Omega_{\ne i}=\bigotimes_{j\ne i}\Omega_j.
\]
Expanding \cref{def:first-order-part} gives
\begin{equation}
\begin{aligned}
  X_1
  &=
  \sum_{i=1}^n
  \left(\prod_{j\ne i}c_j^2\right)
  \Omega_{\ne i}\otimes(\rho_i-\sigma_i)\\
  &=
  \left(\prod_{i=1}^n c_i^2\right)Y.
\end{aligned}
  \label{eq:x1-product-form}
\end{equation}

By \cref{lem:uniform-bures-data},
\[
  \Omega_i\succeq\frac{\mu}{4}\Id.
\]
Moreover, $c_i^2\le1$, so
\[
  \max_i\norm{x_i}_1
  \ge
  \max_i\norm{\rho_i-\sigma_i}_1
  \ge\lambda.
\]
Since partial trace gives $\Tr_{\ne i}Y=x_i$, we also have
\[
  \norm Y_1\ge\max_i\norm{x_i}_1\ge\lambda.
\]

Set
\[
  \eta=\frac{\beta\lambda}{40}.
\]

Choose a rational
\[
  0<\varepsilon\le
  \min\left\{\frac{\mu}{8},\frac{\lambda}{8},
                   \frac{\eta}{4n^2}\right\}.
\]
Apply the constructive algorithm of \cref{lem:uniform-bures-data} with
sufficiently many additional working bits.  If $\widehat A_i,\widehat B_i$
are the resulting rational approximations, set
\[
  \widehat Z_i=\widehat A_i+\widehat B_i,
  \qquad
  \widetilde\Omega_i
  =\frac{\widehat Z_i\widehat Z_i^*}
         {\Tr(\widehat Z_i\widehat Z_i^*)}.
\]
Thus $\widetilde\Omega_i$ is an exact rational density matrix.  Compute a
dyadic $\widehat f_i$ with
$\abs{\widehat f_i-f_i}\le\varepsilon/4$, put
$f_i^+=\min\{1,\widehat f_i+\varepsilon/4\}$, and set
\[
  \widetilde x_i
  =\frac{\rho_i-\sigma_i}{(1+f_i^+)/2}.
\]
Then $\widetilde x_i$ is rational, Hermitian, and traceless.  Increasing the
working precision by $O(\log m+\log(1/\mu))$ bits, the estimates in the
constructive proof give
\[
  0\le f_i\le f_i^+\le1,
  \qquad
  \norm{\widetilde\Omega_i-\Omega_i}_1\le\varepsilon,
  \qquad
  \norm{\widetilde x_i-x_i}_1\le4\varepsilon,
  \qquad
  \norm{x_i}_1,\norm{\widetilde x_i}_1\le4.
\]
Consequently, the lower bound in \cref{lem:uniform-bures-data}, Weyl's
inequality, and tensor-product telescoping show that the rational operator
\[
  \widetilde Y
  =
  \sum_i\widetilde\Omega_{\ne i}\otimes\widetilde x_i
\]
satisfies
\[
  \widetilde\Omega_i\succeq\frac{\mu}{8}\Id,
  \qquad
  \max_i\norm{\widetilde x_i}_1\ge\frac{\lambda}{2},
  \qquad
  \norm{\widetilde Y-Y}_1\le\eta.
\]
Indeed, the telescoping error is at most $4n^2\varepsilon$.  Since
$\eta=\beta\lambda/40$, the required precision, all rational encoding
lengths, and the running time are polynomial in the stated parameters.

Apply \cref{thm:degree-one} to $\widetilde Y$ with
\[
  \gamma=\frac{\mu}{8},
  \qquad
  \lambda'=\frac{\lambda}{2},
  \qquad
  \beta'=\frac{\beta}{4}.
\]
It returns a rational $\widetilde N$ satisfying
\[
  \norm{\widetilde Y}_1
  \le\widetilde N
  \le
  \left(\theta+\frac{\beta}{4}\right)
  \norm{\widetilde Y}_1.
\]
Set
\[
  \widehat N_Y=\widetilde N+\eta.
\]
Since $\norm Y_1\ge\lambda$, $\theta<7$, and $\eta=\beta\lambda/40$, the preceding bounds give
\begin{equation}
  \norm Y_1
  \le\widehat N_Y
  \le
  \left(\theta+\frac{\beta}{2}\right)\norm Y_1.
  \label{eq:certified-Y-estimate}
\end{equation}

Finally, certified approximations to the local fidelities give a rational upper enclosure $\widehat c_\otimes$ satisfying
\begin{equation}
  \prod_i c_i^2
  \le\widehat c_\otimes
  \le
  \left(1+\frac{\beta}{2}\right)\prod_i c_i^2.
  \label{eq:certified-c-product}
\end{equation}
Return
\[
  \widehat N_1=\widehat c_\otimes\widehat N_Y.
\]
Using \eqref{eq:x1-product-form},
\[
\begin{aligned}
  \norm{X_1}_1
  &\le\widehat N_1\\
  &\le
  \left(1+\frac{\beta}{2}\right)
  \left(\theta+\frac{\beta}{2}\right)
  \norm{X_1}_1\\
  &\le
  (1+\beta)(\theta+\beta)\norm{X_1}_1.
\end{aligned}
\]
This proves \eqref{eq:full-rank-X1-interface}.
The required precision, the encoding lengths of the rational local data, and the running time are polynomial in the stated parameters.
\end{proof}

\subsection{Overall Guarantee of the Algorithm}
\label{sec:overall}

\begin{theorem}
\label{thm:main}
There is a deterministic algorithm that, on rational $m\times m$ local density matrices, runs in $\poly(n,m,B)$ time and outputs a rational $\widehat D$ satisfying
\begin{equation}
  \frac{1}{4.73}\D(\rho,\sigma)
  \le\widehat D
  \le 4.73\,\D(\rho,\sigma).
  \label{eq:main-factor}
\end{equation}
The output is zero exactly when $\rho=\sigma$.
\end{theorem}

\begin{proof}
We follow the algorithm construction and analysis roadmap at the beginning of this section.

\noindent \underline{\textit{Step 1: zero detection and regularization.}}
By \cref{lem:zero-scale}, if $\rho_i=\sigma_i$ for every $i$, return zero.
Otherwise, there exist a positive rational $L$ and a site $k$ such that
\[
  2^{-O(B)}
  \le L
  \le \D(\rho_k,\sigma_k)
  \le \D(\rho,\sigma).
\]
Set
\[
  \tau=\frac{L}{2000n},
\]
and define
\[
\begin{aligned}
  \bar\rho_i
  &=(1-\tau)\rho_i+\tau\frac{\Id}{m},
  &
  \bar\sigma_i
  &=(1-\tau)\sigma_i+\tau\frac{\Id}{m},\\
  \bar\rho
  &=\bigotimes_i\bar\rho_i,
  &
  \bar\sigma
  &=\bigotimes_i\bar\sigma_i.
\end{aligned}
\]
Since $2n\tau=L/1000\le\D(\rho,\sigma)/1000$,
\begin{equation}
  \frac{999}{1000}\D(\rho,\sigma)
  \le\D(\bar\rho,\bar\sigma)
  \le\D(\rho,\sigma).
  \label{eq:regularized-distance}
\end{equation}
Both bounds follow from \cref{lem:regularization}.

\noindent \underline{\textit{Step 2: local purification data and the fidelity estimate.}}
For each regularized local pair, use the Uhlmann-optimal purification pair furnished by \cref{lem:uniform-bures-data}, and let $X_1$ be the corresponding first-order part.
Set
\[
  f_i=F(\bar\rho_i,\bar\sigma_i),
  \qquad
  \delta_F=1-\prod_i f_i.
\]
Applying \cref{thm:first-order-bridge} to $(\bar\rho,\bar\sigma)$ gives
\[
\begin{aligned}
  \norm{\bar\rho-\bar\sigma}_1
  &\le\norm{X_1}_1+4\delta_F,\\
  \norm{X_1}_1
  &\le3\norm{\bar\rho-\bar\sigma}_1,\\
  \delta_F
  &\le\frac12\norm{\bar\rho-\bar\sigma}_1.
\end{aligned}
\]

We estimate the two terms in this combined estimator separately.
Fix $\beta=10^{-3}$.
We first estimate $\delta_F$.
Since
\[
  \D(\bar\rho_k,\bar\sigma_k)
  =(1-\tau)\D(\rho_k,\sigma_k)
  \ge\frac L2,
\]
\cref{lem:fvdg} gives
\[
  \delta_F
  \ge1-f_k
  \ge\frac{\D(\bar\rho_k,\bar\sigma_k)^2}{1+f_k}
  \ge\frac{L^2}{8}.
\]
Let
\[
  p=\left\lceil\log_2\frac{16n}{\beta L^2}\right\rceil.
\]
Using \cref{lem:uniform-bures-data}, compute dyadic rationals $\widetilde f_i$ satisfying $\abs{\widetilde f_i-f_i}\le2^{-p}$, and define
\[
  \underline f_i=\max\{0,\widetilde f_i-2^{-p}\},
  \qquad
  \widehat\delta=1-\prod_i\underline f_i.
\]
Then $0\le\underline f_i\le f_i$, and
\[
  \sum_i(f_i-\underline f_i)
  \le n2^{1-p}
  \le\frac{\beta L^2}{8}.
\]
Telescoping the two products therefore gives
\[
  \delta_F
  \le\widehat\delta
  \le\delta_F+\frac{\beta L^2}{8}
  \le(1+\beta)\delta_F.
\]

\noindent \underline{\textit{Step 3: the first-order part.}}
We next estimate $\norm{X_1}_1$.
By \cref{lem:regularization},
\[
  \bar\rho_i,\bar\sigma_i\succeq\frac{\tau}{m}\Id
  \quad\text{for every }i.
\]
Moreover, $L\le1$ implies $\tau<1/2$, and hence
\[
\begin{aligned}
  \max_i\norm{\bar\rho_i-\bar\sigma_i}_1
  &\ge \norm{\bar\rho_k-\bar\sigma_k}_1\\
  &=2(1-\tau)\D(\rho_k,\sigma_k)
  \ge2(1-\tau)L
  >L.
\end{aligned}
\]
Therefore \cref{thm:full-rank-first-order}, applied to $(\bar\rho,\bar\sigma)$ using the same purification pairs and with
\[
  \mu=\frac{\tau}{m},
  \qquad
  \lambda=L,
\]
returns a rational $\widehat N_1$ satisfying
\[
  \norm{X_1}_1
  \le\widehat N_1
  \le(1+\beta)(\theta+\beta)\norm{X_1}_1.
\]

\noindent \underline{\textit{Step 4: assembly.}}
Finally, set
\[
  \widehat T=4\widehat\delta+\widehat N_1.
\]
Then
\[
\begin{aligned}
  \norm{\bar\rho-\bar\sigma}_1
  &\le\norm{X_1}_1+4\delta_F
  \le\widehat T,\\
  \widehat T
  &\le
  4(1+\beta)\delta_F
  +(1+\beta)(\theta+\beta)\norm{X_1}_1\\
  &\le
  (1+\beta)(2+3\theta+3\beta)
  \norm{\bar\rho-\bar\sigma}_1.
\end{aligned}
\]
Together with \eqref{eq:regularized-distance}, this yields
\[
  2\frac{999}{1000}\D(\rho,\sigma)
  \le\widehat T
  \le2(1+\beta)(2+3\theta+3\beta)\D(\rho,\sigma).
\]

Return
\[
  \widehat D=\frac{10583}{100000}\widehat T.
\]
Using $\theta=4+4\log2$ and $\log2<0.693148$,
\[
\begin{aligned}
  \frac{1}{2(999/1000)(10583/100000)}
  &<4.7293<4.73,\\
  2(1+10^{-3})(2+3\theta+3\cdot10^{-3})
  \frac{10583}{100000}
  &<4.7293<4.73.
\end{aligned}
\]
Thus \eqref{eq:main-factor} holds.

Finally, the regularized local inputs and the parameters $\mu=\tau/m$, $\lambda=L$ have polynomial encoding length, and
\[
  \log\frac1\mu=O(B+\log n+\log m),
  \qquad
  \log\frac1\lambda=O(B),
  \qquad
  p=O(B+\log n).
\]
Hence \cref{lem:uniform-bures-data,thm:full-rank-first-order} make the entire computation run in $\poly(n,m,B)$ time.
The output is zero exactly in the initial case detected by \cref{lem:zero-scale}.
\end{proof}

\section{Exact Hardness}
\label{sec:hardness}

In this section, we prove that exactly computing $\D(\rho,\sigma)$ is $\#\mathsf P$-hard.
The following lemma shows that the trace distance between diagonal product states coincides with the total variation distance between the corresponding classical product distributions.

\begin{lemma}[Diagonal embedding]
\label{lem:diagonal-embedding}
Let $P=\bigotimes_iP_i$ and $Q=\bigotimes_iQ_i$ be product distributions on $[m]^n$, and define
\[
  \rho_i=\diag(P_i(1),\ldots,P_i(m)),\qquad
  \sigma_i=\diag(Q_i(1),\ldots,Q_i(m)).
\]
Then
\begin{equation}
  \D(\rho,\sigma)=\TV(P,Q).
  \label{eq:commuting-tv}
\end{equation}
\end{lemma}

\begin{proof}
In the standard product basis, $\rho-\sigma$ is diagonal with entry $P(x)-Q(x)$ at $x\in[m]^n$.
Thus
\[
  \frac12\norm{\rho-\sigma}_1
  =\frac12\sum_x|P(x)-Q(x)|,
\]
which proves \eqref{eq:commuting-tv}.
\end{proof}

The following theorem establishes the $\#\mathsf P$-hardness of exactly computing $\D(\rho,\sigma)$ via a polynomial-time reduction.

\begin{theorem}[Exact computation is hard]
\label{thm:exact-hardness}
For every fixed $m\ge2$, exact computation of $\D(\rho,\sigma)$ is $\#\mathsf P$-hard, even when all local states are diagonal and rational.
\end{theorem}

\begin{proof}
Let $P=\bigotimes_iP_i$ and $Q=\bigotimes_iQ_i$ be product distributions on $\{0,1\}^n$.
Define diagonal qubit states
\[
  \rho_i=\diag(P_i(0),P_i(1)),\qquad
  \sigma_i=\diag(Q_i(0),Q_i(1)).
\]
Then \eqref{eq:commuting-tv} gives
\[
  \D(\rho,\sigma)=\TV(P,Q).
\]
Exact total variation distance for product distributions is $\#\mathsf P$-hard \cite{BhattacharyyaEtAl2025}, so an exact trace-distance oracle solves that problem.
For $m>2$, append zero diagonal coordinates.
\end{proof}

\section{Conclusion}\label{sec-conclusion}

We gave a deterministic $4.73$-approximation for the trace distance between arbitrary rational product quantum states.
The running time is polynomial in the number of sites, the local dimension, and the input bit length.
The main analytic step is a dimension-free comparison between the trace norm of a first-order product operator and a local trace-norm/Jordan-square infimal convolution.
A common regularization makes the associated conic program uniformly solvable in polynomial time.
Exact computation remains $\#\mathsf P$-hard already for diagonal qubit inputs.
The main open question is whether arbitrary noncommuting product inputs admit a polynomial-time relative approximation scheme.

\section*{Acknowledgement}
The proof was obtained with the assistance of GPT-5.5, GPT-5.6 Sol, and Claude Opus 4.8.
GPT-5.6 Sol helped draft the initial version of the manuscript.
The authors posed the problem, guided the development, verified all steps, and accept full responsibility for the content.

\bibliographystyle{alpha}
\bibliography{references}

\clearpage
\appendix

\section{Missing Proofs in
  \texorpdfstring{\cref{sec:prelim}}{the preliminaries}}
\label{app:missing-proofs-prelim}

\subsection{Proof of
  \texorpdfstring{\cref{thm:centered-disjointification}}
    {the centered scalar disjointification theorem}}
\label{app:proof-centered-disjointification}

We first list the following version of the Khinchine inequality.
\begin{lemma}[Sharp $L_1$ Khintchine inequality~~\cite{Szarek1976}]
\label{lem:sharp-l1-khintchine}
Let $\varepsilon_1,\ldots,\varepsilon_n$ be independent Rademacher random variables.
Then, for every $x_1,\ldots,x_n\in\mathbb R$,
\begin{equation}
  \mathbb E_\varepsilon
  \left|\sum_{i=1}^n\varepsilon_i x_i\right|
  \ge
  \frac1{\sqrt2}
  \left(\sum_{i=1}^n x_i^2\right)^{1/2}.
  \label{eq:sharp-l1-khintchine}
\end{equation}
The constant $1/\sqrt2$ is optimal.
\end{lemma}

Then, we give the proof of \Cref{thm:centered-disjointification}.

\begin{proof}[Proof of \Cref{thm:centered-disjointification}] 
\smallskip\noindent\underline{\emph{(1) Symmetrization.}}
If $\mathbb E|\sum_i\xi_i|=\infty$, the claimed inequality is immediate, so assume $\mathbb E|\sum_i\xi_i|<\infty$.
Let $\xi_i'$ be mutually independent copies of the $\xi_i$'s, also independent of the original variables.
Set $X_i=\xi_i-\xi_i'$ and $A=\mathbb E|\sum_iX_i|$.
Each $X_i$ is symmetric because exchanging $\xi_i$ and $\xi_i'$ changes its sign without changing its distribution.

We first compare the two $\kappa$-quantities.
Fix a centered decomposition $X_i=u_i+v_i$, and define $a_i=\mathbb E[u_i\mid\xi_i]$ and $b_i=\mathbb E[v_i\mid\xi_i]$.
Since the variables are centered and $\xi_i'$ is independent of $\xi_i$,
\[
  a_i+b_i
  =\mathbb E[X_i\mid\xi_i]
  =\mathbb E[\xi_i-\xi_i'\mid\xi_i]
  =\xi_i-\mathbb E[\xi_i']
  =\xi_i.
\]
Thus $a_i+b_i$ is an admissible decomposition of $\xi_i$.
Conditional expectation preserves the zero means of $u_i,v_i$, and conditional Jensen gives
\[
  \mathbb E|a_i|
  =
  \mathbb E\left|\mathbb E[u_i\mid\xi_i]\right|
  \le
  \mathbb E\,\mathbb E[|u_i|\mid\xi_i]
  =
  \mathbb E|u_i|
\]
and
\[
  \mathbb E b_i^2
  =
  \mathbb E\left(\mathbb E[v_i\mid\xi_i]\right)^2
  \le
  \mathbb E\,\mathbb E[v_i^2\mid\xi_i]
  =
  \mathbb E v_i^2.
\]
In both displays, the inequality is Jensen's inequality for the convex functions $|z|$ and $z^2$, respectively, and the last equality is the tower property.
Therefore the cost of the decomposition $\xi_i=a_i+b_i$ is no larger than the cost of $X_i=u_i+v_i$.
Taking the infimum over all decompositions of the $X_i$'s gives $\kappa((\xi_i))\le\kappa((X_i))$.
Moreover,
\[
\begin{aligned}
  A
  &=\mathbb E\left|\sum_i\xi_i-\sum_i\xi_i'\right|\\
  &\le\mathbb E\left|\sum_i\xi_i\right|
      +\mathbb E\left|\sum_i\xi_i'\right|\\
  &=2\mathbb E\left|\sum_i\xi_i\right|.
\end{aligned}
\]
The inequality is the scalar triangle inequality, and the last equality uses that $\sum_i\xi_i'$ has the same distribution as $\sum_i\xi_i$.
We have proved
\begin{equation}
  \kappa((\xi_i))\le\kappa((X_i)),
  \qquad
  A\le2\mathbb E\left|\sum_i\xi_i\right|.
  \label{eq:js-symmetrization}
\end{equation}

\smallskip\noindent\underline{\emph{(2) Large coordinates.}}
Let $M=\max_i|X_i|$.
Conditional on the magnitudes $(|X_1|,\ldots,|X_n|)$, we may write $X_i=\varepsilon_i|X_i|$, where the $\varepsilon_i$'s are independent Rademacher signs.
Fix the magnitudes and all signs except the one attached to a largest magnitude $M$.
The two possible values of $\sum_iX_i$ under the remaining random sign differ by $2M$.
By the triangle inequality, the sum of their absolute values is at least $2M$, so their average is at least $M$.
Averaging over the remaining signs and then over the magnitudes yields
\begin{equation}
  \mathbb EM\le A.
  \label{eq:js-max}
\end{equation}
If $A=0$, then $\mathbb EM=0$, so $M=0$ almost surely and hence every $X_i=0$ almost surely.
In that case $\kappa((X_i))=0$ and the theorem is immediate.
We may therefore assume $A>0$ and set $t=2A$.

For every $s\ge t$, Markov's inequality and \eqref{eq:js-max} give $\Pr(M>s)\le\mathbb EM/s\le A/(2A)=1/2$.
Define $p_i(s)=\Pr(|X_i|>s)$ and $p(s)=\sum_ip_i(s)$.
Since the $X_i$'s are independent,
\[
  1-\Pr(M>s)
  =\Pr(|X_i|\le s\text{ for every }i)
  =\prod_i(1-p_i(s)).
\]
Using $1-u\le e^{-u}$ for $u\in[0,1]$, we further obtain
\[
  \prod_i(1-p_i(s))
  \le
  \prod_i e^{-p_i(s)}
  =
  e^{-p(s)}.
\]
On the other hand, $\Pr(M>s)\le1/2$ implies $\prod_i(1-p_i(s))\ge1/2$.
Combining the last two inequalities gives $e^{-p(s)}\ge1/2$, or $p(s)\le\log2$.
Also, $\Pr(M>s)=1-\prod_i(1-p_i(s))\ge1-e^{-p(s)}$.
By concavity, $1-e^{-u}$ lies above the chord joining its values at $0$ and $\log2$.
Hence $1-e^{-u}\ge u/(2\log2)$ on this interval, and therefore $p(s)\le2\log2\,\Pr(M>s)$ for $s\ge t$.

Now truncate each variable at level $t$ by setting $U_i=X_i\mathbf 1_{\{|X_i|>t\}}$ and $V_i=X_i\mathbf 1_{\{|X_i|\le t\}}$.
Then $X_i=U_i+V_i$.
Both pieces are centered because $X_i$ is symmetric and the truncation depends only on $|X_i|$.
For each $i$, the layer-cake identity applied directly to $|X_i|$ gives
\[
  \mathbb E|U_i|
  =
  \mathbb E\!\left[|X_i|\mathbf 1_{\{|X_i|>t\}}\right]
  =
  t\Pr(|X_i|>t)+\int_t^\infty\Pr(|X_i|>s)\,ds.
\]
This follows by writing
\[
  |X_i|\mathbf 1_{\{|X_i|>t\}}
  =
  t\mathbf 1_{\{|X_i|>t\}}
  +\int_t^\infty\mathbf 1_{\{|X_i|>s\}}\,ds
\]
and applying Tonelli's theorem to exchange the order of expectation and integration.
Summing over $i$ and using $p(s)\le2\log2\,\Pr(M>s)$ gives
\begin{equation}
\begin{aligned}
 \sum_i\mathbb E|U_i|
 &=t\,p(t)+\int_t^\infty p(s)\,ds\\
 &\le
  2\log2\left(
    t\Pr(M>t)+\int_t^\infty\Pr(M>s)\,ds
  \right)\\
 &=2\log2\,
   \mathbb E\!\left[M\mathbf 1_{\{M>t\}}\right]\\
 &\le2\log2\,\mathbb EM\\
 &\le2\log2\,A.
\end{aligned}
\label{eq:js-head}
\end{equation}

\smallskip\noindent\emph{\underline{(3) Bounded coordinates.}}
If $\sum_i\mathbb EV_i^2<t^2$, then $(\sum_i\mathbb EV_i^2)^{1/2}<t=2A$.

Suppose now that $\sum_i\mathbb EV_i^2\ge t^2$.
The variables $V_i$ are independent because each is a function of $X_i$.
Expanding the square and using independence for the cross terms gives
\[
\begin{aligned}
  \mathbb E\left(\sum_iV_i^2\right)^2
  &=\sum_i\mathbb EV_i^4
    +2\sum_{i<j}\mathbb EV_i^2\,\mathbb EV_j^2\\
  &=\left(\sum_i\mathbb EV_i^2\right)^2
    +\sum_i\left(
      \mathbb EV_i^4-(\mathbb EV_i^2)^2
    \right)\\
  &\le\left(\sum_i\mathbb EV_i^2\right)^2+\sum_i\mathbb EV_i^4\\
  &\le\left(\sum_i\mathbb EV_i^2\right)^2
      +t^2\sum_i\mathbb EV_i^2\\
  &\le2\left(\sum_i\mathbb EV_i^2\right)^2.
\end{aligned}
\]
The first inequality drops the nonpositive terms $-(\mathbb EV_i^2)^2$.
The second uses $|V_i|\le t$, which implies $V_i^4\le t^2V_i^2$ pointwise.
The last uses $t^2\le\sum_i\mathbb EV_i^2$.

H\"older's inequality and the preceding second-moment bound give
\begin{equation}
  \mathbb E\sqrt{\sum_iV_i^2}
  \ge
  \frac{\left(\sum_i\mathbb EV_i^2\right)^{3/2}}
       {\left(\mathbb E(\sum_iV_i^2)^2\right)^{1/2}}
  \ge
  \frac1{\sqrt2}\left(\sum_i\mathbb EV_i^2\right)^{1/2}.
\label{eq:js-holder}
\end{equation}
Conditioning on all magnitudes $(|X_i|)_i$, \cref{lem:sharp-l1-khintchine} gives
\[
\begin{aligned}
  A
  &=\mathbb E\left|\sum_iX_i\right|\ge
  \frac1{\sqrt2}\,
  \mathbb E\sqrt{\sum_iX_i^2}\ge
  \frac1{\sqrt2}\mathbb E\sqrt{\sum_iV_i^2}.
\end{aligned}
\]
The last inequality uses $V_i^2\le X_i^2$ pointwise.
Together with \eqref{eq:js-holder}, this gives $(\sum_i\mathbb EV_i^2)^{1/2}\le2A$.
The same bound was already proved when $\sum_i\mathbb EV_i^2<t^2$.

\smallskip\noindent\emph{\underline{(4) Putting the pieces together.}}
The decomposition $X_i=U_i+V_i$ is admissible in the definition of $\kappa((X_i))$.
Therefore
\[
\begin{aligned}
  \kappa((X_i))
  &\le
  \sum_i\mathbb E|U_i|
  +\left(\sum_i\mathbb EV_i^2\right)^{1/2}\\
  &\le(2+2\log2)A.
\end{aligned}
\]
Finally, \eqref{eq:js-symmetrization} gives
\[
  \kappa((\xi_i))
  \le\kappa((X_i))
  \le(2+2\log2)A
  \le\theta\,\mathbb E\left|\sum_i\xi_i\right|,
\]
as claimed.
\end{proof}

\section{Missing Proofs in
  \texorpdfstring{\cref{sec:constant}}
    {the constant-factor approximation section}}
\label{app:missing-proofs-constant}

\subsection{Proof of
  \texorpdfstring{\cref{lem:zero-scale}}{the zero-scale lemma}}
\label{app:proof-zero-scale}

\begin{proof}
Since every density matrix has trace one, taking the partial trace over all sites except $i$ gives
\[
  \Tr_{[n]\setminus\{i\}}(\rho)=\rho_i,
  \qquad
  \Tr_{[n]\setminus\{i\}}(\sigma)=\sigma_i.
\]
Consequently, if $\rho=\sigma$, then $\rho_i=\sigma_i$ for every $i$.
The converse is immediate: if all local pairs are equal, then their tensor products are equal.

Now suppose that $\rho\ne\sigma$.
By the choice of $k$, there are indices $a,b$ and $\chi\in\{\Re,\Im\}$ such that, with $\Delta_k=\rho_k-\sigma_k$,
\[
  L=\frac12\abs{\chi((\Delta_k)_{ab})}>0.
\]
The input entries are rational, so $L$ is rational and can be computed in polynomial time by scanning the local matrix entries.

Trace-distance contractivity under partial trace and the inequalities $\norm{\Delta_k}_1\ge\norm{\Delta_k}_{\op} \ge\abs{(\Delta_k)_{ab}}$ give
\begin{align}\label{eq-drhosigmagedrhoisimgai-gel}
  \D(\rho,\sigma)
  \ge
  \D(\rho_k,\sigma_k)
  =\frac12\norm{\Delta_k}_1
  \ge\frac12\abs{(\Delta_k)_{ab}}
  \ge\frac12\abs{\chi((\Delta_k)_{ab})}
  =L.
\end{align}

Finally, the nonzero rational component $\chi((\Delta_k)_{ab})$ is the difference of two $B$-bit rationals.
Putting them over a common denominator shows that its absolute value, and hence $L$, is at least $2^{-O(B)}$.
Together with \eqref{eq-drhosigmagedrhoisimgai-gel}, this proves \eqref{eq-twominusobleql-leqtd}.
\end{proof}

\subsection{Proof of
  \texorpdfstring{\cref{lem:regularization}}{the regularization lemma}}
\label{app:proof-regularization}

\begin{proof}
We first bound the effect of regularization at a single site.
From the definition,
\[
  \rho_i-\bar\rho_i
  =\tau\left(\rho_i-\frac{\Id}{m}\right).
\]
Consequently,
\[
\begin{aligned}
  \D(\rho_i,\bar\rho_i)
  &=\frac{\tau}{2}
    \left\|\rho_i-\frac{\Id}{m}\right\|_1
  \le\frac{\tau}{2}
    \left(\norm{\rho_i}_1+\left\|\frac{\Id}{m}\right\|_1\right)=\tau.
\end{aligned}
\]
The first equality uses homogeneity of the trace norm.
The inequality is the triangle inequality.
The last equality holds because both $\rho_i$ and $\Id/m$ are density matrices, so each has trace norm one.

We now pass from one site to the product state.
For $k=0,\ldots,n$, let
\[
  R_k
  =
  \left(\bigotimes_{j=1}^{k}\bar\rho_j\right)
  \otimes
  \left(\bigotimes_{j=k+1}^{n}\rho_j\right).
\]
Thus $R_0=\rho$ and $R_n=\bar\rho$.
Changing the $k$-th factor gives
\[
  R_{k-1}-R_k
  =
  \left(\bigotimes_{j<k}\bar\rho_j\right)
  \otimes
  \left(\rho_k-\bar\rho_k\right)
  \otimes
  \left(\bigotimes_{j>k}\rho_j\right).
\]
Since $R_0-R_n=\sum_{k=1}^n(R_{k-1}-R_k)$, we obtain
\[
\begin{aligned}
  \D(\rho,\bar\rho)
  &=\frac12\left\|
      \sum_{k=1}^n(R_{k-1}-R_k)
    \right\|_1 \le \frac12\sum_{k=1}^n\norm{R_{k-1}-R_k}_1 =\sum_{k=1}^n
    \frac12\norm{\rho_k-\bar\rho_k}_1\\
  &=\sum_{k=1}^n\D(\rho_k,\bar\rho_k)
  \le n\tau.
\end{aligned}
\]
Repeating the same argument with $\sigma_i$ in place of $\rho_i$ gives $\D(\sigma,\bar\sigma)\le n\tau$.

The triangle inequality and the two product-state bounds give
\[
  \D(\rho,\sigma)-\D(\bar\rho,\bar\sigma)
  \le
  \D(\rho,\bar\rho)+\D(\sigma,\bar\sigma)
  \le2n\tau.
\]
On the other hand, the channel
\[
  \mathcal E_\tau(X)
  =(1-\tau)X+\tau\Tr(X)\frac{\Id}{m}
\]
is CPTP and maps each original local factor to its barred counterpart.
Thus $\bar\rho=\mathcal E_\tau^{\otimes n}(\rho)$ and $\bar\sigma=\mathcal E_\tau^{\otimes n}(\sigma)$, so trace-distance contractivity gives
\[
  \D(\bar\rho,\bar\sigma)\le\D(\rho,\sigma).
\]
Together, these inequalities prove \eqref{eq:regularization-error}.

Finally,
\[
  \bar\rho_i-\frac{\tau}{m}\Id
  =(1-\tau)\rho_i\succeq0,
\]
because $1-\tau>0$ and $\rho_i\succeq0$.
Thus $\bar\rho_i\succeq(\tau/m)\Id$, and the same calculation applies to $\bar\sigma_i$.
\end{proof}

\subsection{Proof of
  \texorpdfstring{\cref{lem:uniform-bures-data}}
    {the uniform Bures data lemma}}
\label{app:proof-uniform-bures-data}

We first record two local numerical lemmas.

\begin{lemma}[Local matrix computations]
\label{lem:local-algebra}
Let $h$ be a rational positive-definite $m\times m$ matrix, and suppose that supplied rational bounds satisfy
\[
  \mu\Id\preceq h\preceq M\Id
  \qquad (0<\mu\le1).
\]
Given $p$, rational approximations to $h^{1/2}$, $h^{-1/2}$, and their products with a fixed number of rational matrices can be computed with operator-norm error at most $2^{-p}$ in time
\begin{equation}
  \poly\!\left(m,B,p,\log\frac1\mu,\log(1+M)\right).
  \label{eq:uniform-spectral-time}
\end{equation}
Here $B$ bounds the encoding lengths of all rational inputs.
\end{lemma}

\begin{proof}
We use a rounded inverse Newton--Schulz iteration.  The underlying exact
iteration is classical; see, for example, \cite{GuoHigham2006}.  Let $s$ be
the least nonnegative integer for which
$R=4^s\ge\max\{1,M\}$, set $X_0=2^{-s}\Id$, and define
\[
  X_{k+1}=\Phi(X_k),
  \qquad
  \Phi(X)=\frac12X(3\Id-hX^2).
\]
Every exact iterate is a polynomial in $h$.  Thus it commutes with $h$, and,
for $E_k=\Id-hX_k^2$, direct multiplication gives
\begin{equation}
  E_{k+1}=\frac14E_k^2(3\Id+E_k).
  \label{eq:newton-residual}
\end{equation}
Since $0\preceq E_0\preceq(1-\mu/R)\Id$, induction in a common eigenbasis
yields
\[
  0\preceq E_k\preceq\Id,
  \qquad
  \norm{E_k}_{\op}
  \le(1-\mu/R)^{2^k}
  \le\exp(-2^k\mu/R).
\]
Moreover $0\preceq X_k\preceq h^{-1/2}$ and
\begin{equation}
  \norm{X_k-h^{-1/2}}_{\op}
  \le\mu^{-1/2}\norm{E_k}_{\op}.
  \label{eq:newton-exact-error}
\end{equation}

We next control rounding.  Let $\mathcal D_q(Y)$ first replace $Y$ by
$(Y+Y^*)/2$, then round the real and imaginary parts of its upper-triangular
entries to $q$ binary fractional places, and finally reflect them across the
diagonal.  Then
\begin{equation}
  \norm{\mathcal D_q(Y)-(Y+Y^*)/2}_{\op}
  \le2m2^{-q}.
  \label{eq:dyadic-rounding}
\end{equation}
Put $M_0=\max\{1,M\}$ and choose
\[
  K=
  \left\lceil\log_2\frac R\mu\right\rceil
  +\left\lceil\log_2\!\left(
     p+\left\lceil\log_2\frac1\mu\right\rceil
      +\left\lceil\log_2M_0\right\rceil+4
   \right)\right\rceil+2.
\]
The preceding exact estimate gives
\[
  \norm{X_K-h^{-1/2}}_{\op}\le\varepsilon,
  \qquad
  \varepsilon=\frac{2^{-p-3}}{M_0}.
\]
On the ball $\norm X_{\op}\le2\mu^{-1/2}$, expanding
$UhU^2-VhV^2$ shows that
\begin{equation}
  \norm{\Phi(U)-\Phi(V)}_{\op}
  \le L\norm{U-V}_{\op},
  \qquad
  L=2+\frac{6M}{\mu}.
  \label{eq:newton-lipschitz}
\end{equation}
Choose $q$ so that $2m2^{-q}L^K\le\varepsilon$, and compute
\[
  \widehat X_0=X_0,
  \qquad
  \widehat X_{k+1}=\mathcal D_q(\Phi(\widehat X_k)).
\]
Hermitianization is an operator-norm contraction relative to the Hermitian
exact iterate.  Hence \eqref{eq:dyadic-rounding} and
\eqref{eq:newton-lipschitz} give
\[
  \norm{\widehat X_k-X_k}_{\op}
  \le2m2^{-q}\sum_{j=0}^{k-1}L^j
  \le\varepsilon
  \qquad(0\le k\le K).
\]
The same induction keeps the rounded iterates in the ball on which the
Lipschitz estimate holds.  Therefore
$\widetilde T=\widehat X_K$ approximates $h^{-1/2}$ with error less than
$2^{-p}$.

To obtain the square root, choose
\[
  q_S=p+3+\left\lceil\log_2m\right\rceil
\]
and set $\widetilde S=\mathcal D_{q_S}(h\widetilde T)$.  Since
$hh^{-1/2}=h^{1/2}$, the pre-rounding error is at most
$2M\varepsilon\le2^{-p-2}$, and the final rounding error is at most
$2m2^{-q_S}\le2^{-p-2}$.

Finally,
\[
  K=O\!\left(\log\frac{1+M}{\mu}
       +\log(p+\log(1/\mu))\right),
  \qquad
  q=p+O(K\log L+\log m+\log M_0).
\]
Each iteration uses a constant number of rational matrix products followed by
dyadic rounding, so all intermediate encoding lengths and the running time
satisfy \eqref{eq:uniform-spectral-time}.  Products with any fixed number of
rational input matrices are handled by increasing the working precision by a
polynomial number of bits, multiplying exactly, and rounding once at the end;
submultiplicativity gives the same $2^{-p}$ error bound.
\end{proof}

\begin{lemma}[Weyl's eigenvalue perturbation bound, Corollary~III.2.6 in \cite{Bhatia1997}]
\label{lem:appendix-weyl-perturbation}
Let $H,K\in\C^{m\times m}$ be Hermitian, and order their eigenvalues as
\[
  \lambda_1(H)\le\cdots\le\lambda_m(H),
  \qquad
  \lambda_1(K)\le\cdots\le\lambda_m(K).
\]
Then, for every $j=1,\ldots,m$,
\[
  \abs{\lambda_j(H)-\lambda_j(K)}
  \le \norm{H-K}_{\op}.
\]
In particular,
\[
  \lambda_{\min}(K)
  \ge
  \lambda_{\min}(H)-\norm{H-K}_{\op}.
\]
\end{lemma}

\begin{proof}[Proof of \Cref{lem:uniform-bures-data}]
Because $a$ and $b$ are density matrices, all their eigenvalues lie in $[0,1]$.
Together with the assumption of the lemma, this gives $\mu\Id\preceq a,b\preceq\Id$.
Define
\[
  C=a^{1/2}ba^{1/2},\qquad
  T=a^{-1/2}C^{1/2}a^{-1/2},
  \qquad
  A=a^{1/2},\quad B=a^{-1/2}C^{1/2}=Ta^{1/2}.
\]
Congruencing the inequality $b\succeq\mu\Id$ by $a^{1/2}$ gives
\[
  C=a^{1/2}ba^{1/2}
  \succeq\mu a
  \succeq\mu^2\Id.
\]
Therefore both $a^{-1/2}$ and $C^{1/2}$ exist.
Moreover,
\[
  T=a^{-1/2}C^{1/2}a^{-1/2}\succeq0.
\]

The matrices $A$ and $B$ are purification matrices for $a$ and $b$ because
\[
  AA^*=a^{1/2}a^{1/2}=a.
\]
\[
\begin{aligned}
  BB^*
  &=a^{-1/2}C^{1/2}C^{1/2}a^{-1/2}
    =a^{-1/2}Ca^{-1/2}=b.
\end{aligned}
\]
Their overlap is positive:
\[
\begin{aligned}
  A^*B
  &=a^{1/2}a^{-1/2}C^{1/2}
    =C^{1/2}\succeq 0.
\end{aligned}
\]
Consequently,
\[
  \Tr(A^*B)
  =\Tr(C^{1/2})
  =\norm{a^{1/2}b^{1/2}}_1
  =F(a,b)=f.
\]
The second equality follows from $\norm X_1=\Tr\sqrt{XX^*}$, applied to
$X=a^{1/2}b^{1/2}$, for which $XX^*=C$.
Thus the two amplitudes attain the fidelity overlap and are optimal by \Cref{lem:uhlmann-matrix}.

The normalization in the definition of $\Omega$ is also explicit:
\[
\begin{aligned}
  \Tr\!\left((A+B)(A+B)^*\right)
  &=\Tr(AA^*)+\Tr(BB^*)+2\Re\Tr(A^*B)\\
  &=1+1+2f=2(1+f).
\end{aligned}
\]
Here $\Tr(AA^*)=\Tr a=1$, $\Tr(BB^*)=\Tr b=1$, and
$\Tr(A^*B)=f$ is real because $A^*B\succeq0$.
Hence $\Omega$ has trace one.

Since $B=Ta^{1/2}$, $A+B=(\Id+T)a^{1/2}$, and therefore
\[
\begin{aligned}
  (A+B)(A+B)^*
  &=(\Id+T)a(\Id+T)\\
  &\succeq\mu(\Id+T)^2\\
  &\succeq\mu\Id.
\end{aligned}
\]
The first inequality is obtained by congruencing $a\succeq\mu\Id$ with
$\Id+T$; the second follows from $T\succeq0$.

Finally, by the Cauchy--Schwarz inequality,
\[
  f=\norm{a^{1/2}b^{1/2}}_1
  \le\norm{a^{1/2}}_2\norm{b^{1/2}}_2
  =\sqrt{\Tr a}\sqrt{\Tr b}
  =1.
\]
Therefore,
\[
  \Omega
  \succeq
  \frac{\mu}{2(1+f)}\Id
  \succeq
  \frac{\mu}{4}\Id,
\]
which proves \eqref{eq:bisector-gap}.

It remains to construct the approximations.  The upper bounds
$a,b\preceq\Id$ imply
\[
  \mu^2\Id\preceq C=a^{1/2}ba^{1/2}\preceq a\preceq\Id.
\]
In particular,
$\norm{a^{-1/2}}_{\op}\le\mu^{-1/2}$ and
$\norm{C^{1/2}}_{\op}\le1$.

Fix $p\ge1$, put $\delta=2^{-p}$, and choose a dyadic
$\varepsilon=2^{-r}$ such that
\begin{equation}
  0<\varepsilon
  \le\frac{\delta\mu^2}{2^{12}m}.
  \label{eq:local-internal-precision}
\end{equation}
The least such $r$ satisfies
$r=p+O(\log m+\log(1/\mu))$ and is found by exact rational comparisons.
Use the explicit construction in the proof of \Cref{lem:local-algebra} for
$a$ to obtain dyadic Hermitian matrices
$\widehat A,\widehat R$ with
\[
  \norm{\widehat A-a^{1/2}}_{\op}\le\varepsilon,
  \qquad
  \norm{\widehat R-a^{-1/2}}_{\op}\le\varepsilon.
\]
Form the exact rational Hermitian positive-semidefinite matrix
\[
  C_0=\widehat A b\widehat A.
\]
Because $\norm{\widehat A}_{\op}\le2$, multiplication and
\eqref{eq:local-internal-precision} give
\begin{equation}
  \norm{C_0-C}_{\op}
  \le3\varepsilon
  \le\frac{\mu^2}{4}.
  \label{eq:C0-error}
\end{equation}
Hence \Cref{lem:appendix-weyl-perturbation} gives
\[
  \frac{3\mu^2}{4}\Id\preceq C_0\preceq4\Id.
\]
We may therefore use the same construction for $C_0$ and obtain a dyadic
Hermitian $\widehat S$ satisfying
$\norm{\widehat S-C_0^{1/2}}_{\op}\le\varepsilon$.

For positive-definite $H,K$, the square-root perturbation bound
\begin{equation}
  \norm{H^{1/2}-K^{1/2}}_{\op}
  \le
  \frac{\norm{H-K}_{\op}}
       {\sqrt{\lambda_{\min}(H)}+\sqrt{\lambda_{\min}(K)}}
  \label{eq:square-root-perturbation}
\end{equation}
is standard; see \cite{Schmitt1992}.  Using \eqref{eq:C0-error} gives
\[
  \norm{\widehat S-C^{1/2}}_{\op}
  \le\varepsilon+\frac{3\varepsilon}{\mu}
  \le\frac{4\varepsilon}{\mu}.
\]
Define the dyadic matrices and scalar
\[
  \widehat B=\widehat R\widehat S,
  \qquad
  \widehat Z=\widehat A+\widehat B,
  \qquad
  \widehat f=\Tr\widehat S.
\]
Since $\norm{\widehat S}_{\op}\le3$, submultiplicativity gives
\begin{equation}
  \norm{\widehat B-B}_{\op}
  \le3\varepsilon
     +\mu^{-1/2}\frac{4\varepsilon}{\mu}
  \le7\varepsilon\mu^{-3/2}.
  \label{eq:Bhat-error}
\end{equation}
Also
\[
  \abs{\widehat f-f}
  \le m\norm{\widehat S-C^{1/2}}_{\op}
  \le\frac{4m\varepsilon}{\mu}.
\]

To approximate the bisector state while retaining an exact rational density
matrix for later use, put
\[
  \widehat G=\widehat Z\widehat Z^*,
  \qquad
  \widehat t=\Tr\widehat G,
  \qquad
  \Omega^{\mathrm{rat}}=\frac{\widehat G}{\widehat t}.
\]
Let $u=8\sqrt m\,\varepsilon\mu^{-3/2}$.  Then
\eqref{eq:Bhat-error} implies
$\norm{\widehat Z-(A+B)}_F\le u$.  Moreover,
$\norm{A+B}_F^2=2(1+f)\in[2,4]$, so
$\norm{\widehat Z}_F\le3$ and the Schatten H\"older inequality gives
\begin{equation}
  \norm{\widehat G-(A+B)(A+B)^*}_1
  \le5u.
  \label{eq:Gram-error}
\end{equation}
The choice \eqref{eq:local-internal-precision} makes $5u<1$; hence
$\widehat t\ge1$, and normalization gives
\begin{equation}
  \norm{\Omega^{\mathrm{rat}}-\Omega}_1
  \le10u<\frac\delta2.
  \label{eq:rational-bisector-error}
\end{equation}
Finally, choose
\[
  q_\Omega=p+3+\left\lceil\log_2m\right\rceil
\]
and set
$\widetilde\Omega=\mathcal D_{q_\Omega}(\Omega^{\mathrm{rat}})$.
Then \eqref{eq:dyadic-rounding} and
\eqref{eq:rational-bisector-error} give
$\norm{\widetilde\Omega-\Omega}_{\op}\le\delta$.

Set
\[
  \widetilde A=\widehat A,
  \qquad
  \widetilde B=\widehat B,
  \qquad
  \widetilde f=\widehat f.
\]
The bounds above and \eqref{eq:local-internal-precision} imply
\[
  \max\left\{
    \norm{\widetilde A-A}_{\op},
    \norm{\widetilde B-B}_{\op},
    \abs{\widetilde f-f},
    \norm{\widetilde\Omega-\Omega}_{\op}
  \right\}
  \le\delta.
\]
All calls to \Cref{lem:local-algebra} use precision
$p+O(\log m+\log(1/\mu))$ and rational spectral bounds of polynomial
encoding length.  All remaining operations are rational matrix products,
traces, normalization, and dyadic rounding.  This proves the stated running
time and completes the proof.
\end{proof}

\end{document}